\PassOptionsToPackage{colorlinks,linkcolor={blue},citecolor={blue},urlcolor={red},breaklinks=true,final}{hyperref}
\PassOptionsToPackage{final}{graphicx}
\documentclass[a4paper,UKenglish,final]{lipics-v2018}

\makeatletter
\def\renewtheorem#1{%
  \expandafter\let\csname#1\endcsname\relax
  \expandafter\let\csname c@#1\endcsname\relax
  \gdef\renewtheorem@envname{#1}
  \renewtheorem@secpar
}

\def\renewtheorem@secpar{\@ifnextchar[{\renewtheorem@numberedlike}{\renewtheorem@nonumberedlike}}
\def\renewtheorem@numberedlike[#1]#2{\newtheorem{\renewtheorem@envname}[#1]{#2}}

\def\renewtheorem@nonumberedlike#1{
  \def\renewtheorem@caption{#1}
  \edef\renewtheorem@nowithin{\noexpand\newtheorem{\renewtheorem@envname}{\renewtheorem@caption}}
  \renewtheorem@thirdpar
}

\def\renewtheorem@thirdpar{\@ifnextchar[{\renewtheorem@within}{\renewtheorem@nowithin}}
\def\renewtheorem@within[#1]{\renewtheorem@nowithin[#1]}
\makeatother

\usepackage{ifdraft}

\ifdraft{
  \usepackage[layout=footnote,draft]{fixme}
}{
  \usepackage[layout=footnote,final]{fixme}
}

\FXRegisterAuthor{sg}{asg}{SG}	
\FXRegisterAuthor{rn}{arn}{RN}	

\theoremstyle{plain}
\newtheorem{proposition}{Proposition}

\theoremstyle{remark}
\renewtheorem{remark}[theorem]{Remark}

\usepackage{microtype} 

\usepackage{caption}
\usepackage{subcaption}
\usepackage{wrapfig}

\usepackage{color}
\usepackage{stmaryrd}

\usepackage{savesym}

\savesymbol{degree}
\savesymbol{leftmoon}
\savesymbol{rightmoon}
\savesymbol{fullmoon}
\savesymbol{newmoon}
\savesymbol{diameter}
\savesymbol{emptyset}
\savesymbol{bigtimes}
\savesymbol{triangleright}

\usepackage{bm}

\usepackage[matha]{mathabx}

\restoresymbol{other}{emptyset}
\restoresymbol{other}{triangleright}

\usepackage{wasysym}
\usepackage{bbold}

\usepackage{todos}

\providecommand{\catname}{\mathbf} 
\providecommand{\clsname}{\mathcal}
\providecommand{\oname}[1]{\operatorname{\mathsf{#1}}}

\def\defcatname#1{\expandafter\def\csname B#1\endcsname{\catname{#1}}}
\def\defcatnames#1{\ifx#1\defcatnames\else\defcatname#1\expandafter\defcatnames\fi}
\defcatnames ABCDEFGHIJKLMNOPQRSTUVWXYZ\defcatnames

\def\defclsname#1{\expandafter\def\csname C#1\endcsname{\clsname{#1}}}
\def\defclsnames#1{\ifx#1\defclsnames\else\defclsname#1\expandafter\defclsnames\fi}
\defclsnames ABCDEFGHIJKLMNOPQRSTUVWXYZ\defclsnames

\def\defbbname#1{\expandafter\def\csname BB#1\endcsname{\mathbb{#1}}}
\def\defbbnames#1{\ifx#1\defbbnames\else\defbbname#1\expandafter\defbbnames\fi}
\defbbnames ABCDEFGHIJKLMNOPQRSTUVWXYZ\defbbnames

\def\Set{\catname{Set}}

\providecommand{\argument}{\operatorname{-\!-}}

\DeclareOldFontCommand{\bf}{\normalfont\bfseries}{\mathbf}
\providecommand{\mplus}{{\scriptscriptstyle\bf+}} 	

\providecommand{\PSet}{{\mathcal P}}				
\providecommand{\Id}{\operatorname{Id}}

\providecommand{\Hom}{\mathsf{Hom}}
\providecommand{\id}{\mathsf{id}}

\providecommand{\comp}{\mathbin{\circ}}

\providecommand{\ito}{\hookrightarrow}						
			
\providecommand{\dar}{\kern-1.2pt\operatorname{\downarrow}}	
\providecommand{\uar}{\kern-1.2pt\operatorname{\uparrow}}	
\providecommand{\mto}{\mapsto}
\providecommand{\xto}[1]{\,\xrightarrow{#1}\,}

\providecommand{\appr}{\sqsubseteq}

\providecommand{\bigjoin}{\bigsqcup}

\providecommand{\brks}[1]{\langle #1\rangle}

\providecommand{\inl}{\oname{inl}}
\providecommand{\inr}{\oname{inr}}
\providecommand{\inj}{\oname{in}}

\DeclareSymbolFont{Symbols}{OMS}{cmsy}{m}{n}
\DeclareMathSymbol{\iobj}{\mathord}{Symbols}{"3B}

\providecommand{\ev}{\oname{ev}}

\usepackage{stmaryrd}

\providecommand{\lsem}{\llbracket}
\providecommand{\rsem}{\rrbracket}
\providecommand{\sem}[1]{\lsem #1 \rsem}

\providecommand{\by}[1]{\text{/\!/~#1}}			
\providecommand{\pacman}[1]{}					

\providecommand{\noqed}{\def\qed{}}				

\providecommand{\mone}{{\text{\kern.5pt\rmfamily-}\sf\kern-.5pt1}}

\makeatletter
\@ifpackageloaded{enumitem}{}{\usepackage[loadonly]{enumitem}}		
\makeatother

\newlist{citemize}{itemize}{1}
\setlist[citemize]{label=\labelitemi,wide} 

\newlist{cenumerate}{enumerate}{1}
\setlist[cenumerate,1]{label=\arabic*.~,ref={\arabic*},wide} 

\makeatletter
\def\mfix#1{\oname{#1}\@ifnextchar\bgroup\@mfix{}}	
\def\@mfix#1{#1\@ifnextchar\bgroup\mfix{}}			
\makeatother

\providecommand{\case}[3]{\mfix{case}{\mathbin{}#1}{of}{#2}{\kern-1pt;}{\mathbin{}#3}}

\usepackage{etoolbox}

\usepackage{tikz}

\usepackage{graphicx}

\renewcommand{\dar}{\operatorname{\raisebox{.3ex}{\reflectbox{\rotatebox[origin=c]{-90}{$\to$}}}\kern1pt}}
\renewcommand{\uar}{\operatorname{\raisebox{.3ex}{\reflectbox{\rotatebox[origin=c]{90}{$\to$}}}\kern1pt}}

\title{A Semantics for Hybrid Iteration}

\titlerunning{A Semantics for Hybrid Iteration}

\author{Sergey Goncharov}{Lehrstuhl f\"ur Theoretische Informatik, Friedrich-Alexander Universit\"at Erlangen-N\"urnberg, Germany}
{sergey.goncharov@fau.de}{}{Research supported by Deutsche Forschungsgemeinschaft (DFG)
under project GO~2161/1-2.}

\author{Julian Jakob}{Lehrstuhl f\"ur Theoretische Informatik, Friedrich-Alexander Universit\"at Erlangen-N\"urnberg, Germany}
{julian.jakob@fau.de}{}{}

\author{Renato Neves}{INESC TEC (HASLab) \& University of Minho, Portugal}
{nevrenato@di.uminho.pt}{}{Research supported by ERDF -- European Regional Development
Fund through the Operational Programme for Competitiveness and
Internationalisation -- COMPETE 2020 Programme and by National Funds
through the Portuguese funding agency, FCT -- Funda\c{c}\~{a}o para a
Ci\^{e}ncia e a Tecnologia within projects POCI-01-0145-FEDER-016692
and 02/SAICT/2017.}

\authorrunning{S. Goncharov and J. Jakob and R. Neves}

\Copyright{Sergey Goncharov and Julian Jakob and Renato Neves}

\subjclass{} 


\category{}

\relatedversion{}

\supplement{}

\funding{}

\EventEditors{Sven Schewe and Lijun Zhang}
\EventNoEds{2}
\EventLongTitle{29th International Conference on Concurrency Theory (CONCUR 2018)}
\EventShortTitle{CONCUR 2018}
\EventAcronym{CONCUR}
\EventYear{2018}
\EventDate{September 4--7, 2018}
\EventLocation{Beijing, China}
\EventLogo{}
\SeriesVolume{118}
\ArticleNo{22}
\nolinenumbers 
\hideLIPIcs  

\usepackage{scalerel}

\usepackage{proof}
\newcommand{\infrule}[2]{\frac{#1}{#2}}
\newcommand{\anonrule}[3]{\infrule{#2}{#3}}
\newcommand{\istar}{\dagger}  	
\newcommand\pstar{{\kern1pt\scalerel*{\includegraphics{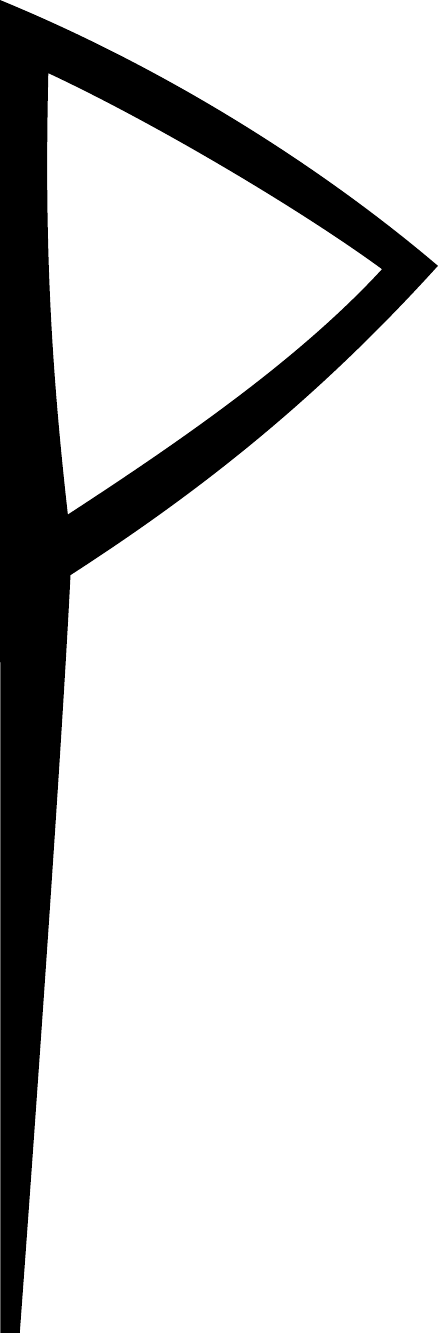}}{\dagger}}}

\newcommand{\iistar}{\ddagger}  
\newcommand{\klstar}{\star}     
\newcommand{\kklstar}{\text{\kreuz}}

\newcommand{\cpto}{
  \mathrel{\raisebox{0.5ex}{\kern3pt\ensuremath{\mathrel{\tikz{ \draw [-stealth,line width=0.4] (0.6ex,1ex) -- (0,1ex) -- (0,0.4ex) -- (2.2ex,0.4ex); }}}\kern3pt}}
}
\providecommand{\dist}{\oname{dist}}

\newcommand{\dr}{{\operatorname{\mathsf{d}}}} 
\renewcommand{\ev}{{\operatorname{\mathsf{e}}}} 

\newcommand{\ite}[3]{\mfix{\kern-2pt}{#1}{\lhd}{#2}{\rhd}{\mathbin{}#3}}

\newcommand{\Reals}{\BBR} 
\newcommand{\Rz}{\BBR_{\mplus}} 
\newcommand{\nats}{\mathbb{N}} 
\newcommand{\Rze}{\overline{\BBR}_{\mplus}} 
\newcommand{\blank}{\, - \,}
\newcommand{\const}[1]{\underline{#1}}

\renewcommand{\comp}{\,}
\newcommand{\pAt}{\mathsf{At}}
\newcommand{\End}{\mathsf{End}}
\newcommand{\prog}[1][p]{\mathsf{#1}}
\newcommand{\sComp}{\hspace{1pt}\large{\mathbf{;}}\hspace{2pt}}
\newcommand{\pv}[1]{\langle #1 \rangle}
\newcommand{\img}{\mathsf{Img}}

\renewcommand{\BBH}{{\bm{\mathsf{H}}}}
\renewcommand{\BBM}{{\bm{\mathsf{M}}}}
\renewcommand{\BBT}{{\bm{\mathsf{T}}}}

\usepackage{nccmath}

\newcommand{\dom}{\oname{dom}}
\renewcommand{\sup}{\oname{sup}}

\usepackage{tikz-cd}

\usepackage{pgfplots}

\pgfplotsset{
  xmin=0,
  xmax=1,
  ymax=1,
  grid=both,
  minor tick num=4,
  grid style={line width=.1pt, draw=gray!10},
  major grid style={line width=.2pt,draw=gray!50},
  axis lines=middle,
  axis line style={-Stealth},
  tick label style={font=\tiny},
  enlargelimits={abs=0.05}
}

\usetikzlibrary{
  fit,%
  calc,%
  arrows,%
  arrows.meta,%
  intersections,%
  shapes.misc,
  shapes.arrows,%
  patterns,%
  automata,%
  chains,%
  matrix,%
  positioning,
  scopes,%
  decorations.markings,%
  decorations.pathmorphing,%
  external
}

\tikzset{
  commutative diagrams/.cd,
  arrow style=tikz,
  diagrams={>=stealth},
  row sep=large,
  column sep = huge
}

\begin{document}

\maketitle

\begin{abstract}
  The recently introduced notions of \emph{guarded traced (monoidal)
    category} and \emph{guarded (pre-)iterative monad} aim at unifying
  different instances of partial iteration whilst keeping in touch
  with the established theory of total iteration and preserving its
  merits. In this paper we use these notions and the corresponding
  stock of results to examine different types of iteration for hybrid
  computation. As a starting point we use an available notion of
  \emph{hybrid monad} restricted to the category of sets, and modify
  it in order to obtain a suitable notion of guarded iteration with
  guardedness interpreted as \emph{progressiveness} in time -- we
  motivate this modification by our intention to capture \emph{Zeno
    behaviour} in an arguably general and feasible way. We illustrate
  our results with a simple programming language for hybrid
  computation which is interpreted over the developed semantic
  foundations.
\end{abstract}

\section{Introduction}


Iteration is a basic concept of computer science that takes different
forms across numerous strands, from formal languages, to process
algebras and denotational semantics. From a categorical point of view, using the
definite perspective of Elgot~\cite{Elgot75}, iteration is an operator
\begin{ceqn}
\begin{align}
  \label{eq:elgot-iter}
  \anonrule{}{f:X\to Y+X}{f^{\istar}:X\to Y}
\end{align}
\end{ceqn}
\noindent
that runs the function $f$ and terminates if the result is in~$Y$,
otherwise it proceeds with the result repetitively.
%
One significant difficulty in the unification of various forms of
iteration is that the latter need not be total, but can be defined only
for a certain class of morphisms whose definition depends on the
nature of the specific example at hand.  In process algebra, for
example, one typically considers recursive solutions of \emph{guarded}
process definitions, in complete metric spaces only fixpoints of
\emph{contractive} maps (which can then be found uniquely thanks to
\emph{Banach's fixpoint theorem}), and in domain theory only least
fixpoints over \emph{pointed predomains} (i.e.\ \emph{domains}). These
examples have recently been shown as instances of the unifying notion of
\emph{guarded traced
  category}~\cite{GoncharovSchroderEtAl17,GoncharovSchroder18}.

In this work we aim to extend the stock of examples of this notion by
including iteration on \emph{hybrid computation}, which are encoded
in the recently introduced hybrid
monad~\mbox{\cite{NevesBarbosaEtAl16,neves18}}. We argue that in the
hybrid context guardedness corresponds to \emph{progressiveness} --
the property of trajectories to progressively extend over time during the
iteration process (possibly converging to a finite trajectory in the
limit) -- we illustrate and examine the corresponding iteration operator
and use it to develop while-loops for hybrid denotational semantics.


\begin{wrapfigure}[13]{r}{0.45\textwidth}
\vspace{0pt}
  \centering
\begin{tikzpicture}
\begin{axis}[xmin=0, xmax=4, ymin=0, ymax=1.1,width=.45\textwidth]

\def\gconst{9.8}
\def\cconst{.8}
\def\xconst{1}

\foreach \x in {0, ..., 20} {
  \def\lp{(1 - \cconst^\x) * (1 + \cconst) * sqrt(2 * \xconst / \gconst) / (1 - \cconst)}
  \def\mp{\lp +  \cconst^\x * sqrt(2 * \xconst / \gconst)}
  \def\rp{(1 - \cconst^(\x+1)) * (1 + \cconst) * sqrt(2 * \xconst / \gconst) / (1 - \cconst)}

  \edef\temp{\noexpand\addplot[blue, thick, smooth, domain=\lp:\mp, samples=10] {\xconst * \cconst^(2 * \x) -\gconst * (x - \lp)^2 /2};}\temp
  \edef\temp{\noexpand\addplot[blue, thick, smooth, domain=\mp:\rp, samples=10] { - (2 * \cconst + 1) * \xconst * \cconst^(2 * \x) + (\cconst + 1) * (x - \lp) * sqrt(2 * \gconst * \xconst * \cconst^(2 * \x)) -\gconst * (x - \lp)^2 /2};}\temp
}
\end{axis}
\end{tikzpicture}
\vspace{-0pt}
  \caption{Bouncing ball's movement.}\label{fig:bouncing}
\vspace{-0pt}
\end{wrapfigure}
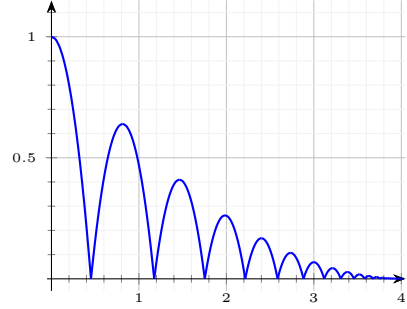

Hybrid computation is inherent to systems that combine
discrete and continuous, \emph{physical
  behaviour}~\cite{Tabuada09,Platzer10, Alur15}.
Traditionally qualified as \emph{hybrid} and born in the context of
control theory~\cite{witsenhausen66}, they range from
computational devices interacting with their physical, external
environment to chemical/biological reactions and physical processes
that are subjected to discrete changes, such as combustions and
impacts. Typical examples include pacemakers, cellular division
processes, cruise control systems, and electric/water grids. Let us
consider, for example, the following \emph{hybrid program}, written in
an algebraic programming style, and with
$(\dot{\mathtt{x}} = \mathtt{t} \> \& \> \mathtt{r})$ denoting `let variable
$\mathtt{x}$ evolve according to~$\mathtt{t}$ during $\mathtt{r}$ milliseconds'.
\begin{ceqn}
\begin{align*}
  \prog[(\dot{v} = 1 \> \& \> 1)  +_{v \leq 120}
  (\dot{v} = -1 \> \& \> 1)]
\end{align*}
\end{ceqn}
\noindent
It represents a (simplistic) cruise controller that either accelerates
$(\dot{\mathtt{v}} = 1 \> \& \> 1)$ or brakes
$(\dot{\mathtt{v}} = -1 \> \& \> 1)$ during one millisecond depending
if the car's velocity $\mathtt{v}$ is lower or greater than
120km/h. This program naturally fits in a slightly more
sophisticated scenario obtained by wrapping a non-terminating
while-loop around it:
\begin{ceqn}
  \begin{align}\label{eq:while-example}
  \mathsf{while}  \> \mathsf{true} \> \{
  (\dot{\mathtt{v}} = 1 \> \& \> 1)  +_{\mathsf{v} \leq 120}
  (\dot{\mathtt{v}} = -1 \> \& \> 1) \}
\end{align}
\end{ceqn}
%
Now the resulting program runs \emph{ad infinitum}, measuring the
car's velocity every millisecond and changing it as specified by the
if-then-else condition. \emph{How should we systematically interpret
  such while-loops?}

Iteration on hybrid computation is notoriously difficult to handle
due to the so called \emph{Zeno behaviour}
\cite{hofner_phdthesis,ames05,zhang2001}, a phenomenon of unfolding an
iteration loop infinitely often in finite time, akin to the scenarios
famously described by the greek philosopher Zeno, further analysed by
Aristotle~\cite[Physics, 
231a–241b]{Aristotle08},
and since then by many
others. To illustrate this, consider a bouncing ball dropped at a
positive height and with no initial velocity. Due to the gravitational
acceleration $\mathsf{g}$, it falls into the ground and bounces back
up, losing a portion of its kinetic energy. In order to model this
system, one can start by writing the program,
\begin{ceqn}
\begin{align}
\label{bouncing:ball}
 (\dot{\mathsf{p}} = \mathsf{v}, \dot{\mathsf{v}} = \mathsf{g} \> \& \> \mathsf{p}
 \leq 0 \wedge \mathsf{v} \leq 0) ;  (\mathsf{v} := \mathsf{v} \times -0.5)
\end{align}
\end{ceqn}
\noindent
to specify the (continuous) change of height $\mathsf{p}$, and also the
(discrete) change of velocity $\mathsf{v}$ when the ball touches the
ground; the expression $\mathsf{p} \leq 0 \wedge \mathsf{v} \leq 0$
provides the termination condition: the ball stops when both its height and
velocity do not exceed zero. Then,
abbreviating program~\eqref{bouncing:ball} to $\mathsf{b}$, one
writes,
\begin{ceqn}
\begin{align*}
  (\mathsf{p := 1, v := 0}) ;
  \underbrace{\mathsf{b ;  \dots ;  b}}_{\text{n times}}
\end{align*}
\end{ceqn}
as the act of dropping the ball and letting it bounce exactly n times.
One may also wish to drop the ball and let it bounce \emph{until it
stops} (see Fig\,\ref{fig:bouncing}), using some form of infinite
iteration on~$\mathsf{b}$ and thus giving rise to  Zeno
behaviour. 
Only a few existing approaches aim to systematically work with Zeno
behaviour, e.g.\ in~\cite{hofner_phdthesis,HofnerMoller11} this is
done by relying on non-determinism, although the results seem to
introduce undesirable behaviour in some occasions (see details in the
following subsection).  Here, we do regard Zeno behaviour as an
important phenomenon to be covered and as such helping to design and
classify notions of iteration for hybrid semantics in a systematic and
compelling way.


\subsection{Related Work, Contributions, Roadmap, and Notation}
There exist two well-established program semantics for hybrid systems:
H\"ofner's `\emph{Algebraic calculi for hybrid systems}'
\cite{hofner_phdthesis} where programs are interpreted as sets of
trajectories, and Platzer's Kleene algebra \cite{Platzer10}
interpreting programs as maps $X \to \PSet X$ for the powerset
functor~$\PSet $. Both approaches are inherently non-deterministic and the
corresponding iteration operators crucially rely on
non-determinism. In~\cite{Platzer10}, the iteration
operator is modelled  by the Kleene star $(\argument)^\ast$, i.e.\ essentially by the
non-deterministic choice between all possible \emph{finite} iterates
of a given program $\mathsf{p}$; more formally, $\mathsf{p}^\ast$ is
the \emph{least fixpoint} of
\begin{ceqn}
\begin{align*}
    x \mapsto \mathsf{p} ;  x + \mathsf{skip}
\end{align*}
\end{ceqn}
\noindent
Semantics based on Kleene star deviates from the (arguably more
natural) intuition given above for the non-terminating
while-loop~\eqref{eq:while-example}.
It is also possible to extend the non-determinis\-tic perspective
summarised above to a more abstract setting via a monad that combines
hybrid computations and non-determinism~\cite{dahlqvist18}, but in the
present work we restrict ourselves to a \emph{purely hybrid setting},
in order to study genuinely hybrid computation in isolation, without
being interfered with other computational effects such as
non-determinism.
%

One peculiarity of the  Kleene star in~\cite{Platzer10} is that it is rather
difficult to use for modelling programs with Zeno behaviour, the
problem the authors are confronted with in \cite{hofner_phdthesis,HofnerMoller11}.
The authors of op.cit.\, extend the Kleene star setting with an
\emph{infinite iteration} operator $(\argument)^\omega$ that for a given
program~$\mathsf{p}$ returns the \emph{largest fixpoint} of the
function
\begin{ceqn}
\begin{align*}
    x \mapsto \mathsf{p} ;  x
\end{align*}
\end{ceqn}
on programs. As argued in~\cite{hofner_phdthesis,HofnerMoller11}, this operator still does not adequately
capture the semantics of hybrid iteration, as it yields `too much
behaviour', e.g.\ if $\mathsf{p}=\mathsf{skip}$, $\mathsf{p}^\omega$
is the program containing all trajectories while we are expecting it
to be $\mathsf{skip}$. This is fixed by combining various
techniques for obtaining a desirable set of behaviours, but unexpected
behaviour could still appear at the smallest instant of time that is
not reached by finite iterations
\cite{hofner_phdthesis,HofnerMoller11}. For the bouncing ball, this
entails that at the instant in which it is supposed to stop, it can
appear below ground or shoot up to the
sky. 

Other types of formalisms for hybrid systems were proposed in the last
decades, including e.g.\ the definite case of hybrid automata
\cite{hybridautomata}, whose distinguishing feature is the ability of
state variables to evolve continuously, and Hybrid \textsc{CSP}
\cite{hcsp}, an extension of \textsc{CSP} by expressions with time
derivatives. More recently, an elegant specification language handling
continuous behaviour of hybrid systems via \emph{non-standard
  analysis} was introduced in ~\cite{suenaga11}.

\smallskip
\noindent
\textbf{Contributions.} We propose semantic foundations for (Elgot)
iteration in a hybrid setting: we identify two new monads for hybrid
computations, one of which supports a partial guarded iteration
operator, characterized as a least solution of the corresponding
fixpoint equation, and another one extending the first and carrying a
total iteration operator, although not generally being characterized
in an analogous way. We show that both operators do satisfy the
standard equational principles of iteration
theories~\cite{BloomEsik93,Elgot75} together with
uniformity~\cite{SimpsonPlotkin00}. Moreover, we develop a language
for hybrid computation with full-fledged while-loops as a prominent
feature and interpret it using the underlying monad-based
semantics. We discuss various use case scenarios and demonstrate
various aspects of the iterative behaviour.

\smallskip
\noindent
\textbf{Plan of the paper.} We proceed by defining a simple
programming language for hybrid computation in
Section~\ref{sec:language}, in order to present and discuss challenges
related to defining a desirable semantics for it. In
Section~\ref{sec:guarded} we provide a summary of guarded (Elgot)
iteration theory. In Sections~\ref{sec:fist} and~\ref{sec:iter} we
present our main technical developments, including two new monads
$\BBH_\mplus$ and $\BBH$ for hybrid computation and the corresponding
iteration operators. In Section~\ref{sec:sem} we provide a semantics
for the while-loops of our programming language and then conclude in
Section~\ref{sec:concl}.

All omitted proofs can be found in the paper's appendix.

\smallskip
\noindent
\textbf{Notation.} 
We assume basic familiarity with the language of category theory
\cite{MacLane71}, monads \cite{MacLane71,Awodey2010}, and topology
\cite{Eng89}. Some conventions regarding notation are in
order. By $|\BC|$ we denote the class of objects of a category $\BC$
and by $\Hom_{\BC}(A,B)$ ($\Hom(A,B)$, if no confusion arises) the
set of morphisms $f:A\to B$ from $A\in |\BC|$ to $B\in|\BC|$. We
denote the set of \emph{Kleisli endomorphisms} $\Hom_{\BC}(X,TX)$ by $\End_{\BBT}(X)$.
We agree to omit indices at natural transformations. We identify monads with the
corresponding Kleisli triples, and use blackboard characters to refer
to a monad and the corresponding roman letter to the monad's
functorial part, e.g.\ $\BBT=(T,\eta,(\argument)^\klstar)$ denotes a
monad over a functor $T$ with $\eta:\Id\to T$ being the \emph{unit}
and $(\argument)^\klstar:\Hom(X,TY)\to \Hom(TX,TY)$ being the
corresponding \emph{Kleisli lifting}. Most of the time we work in the
category $\Set$ of sets and functions.
We write $\Rz$ and $\Rze$ for the sets of non-negative reals, and
non-negative reals extended with infinity~$\infty$
respectively. Given $e:\Rz\to X$ and $t\in\Rz$, we
denote by $e^t$ the application $e(t)$. Given $x\in X$,
$\const x:Y\to X$ is the function constantly equal to $x$. We
use if-then-else constructs of the form $\ite{p}{b}{q}$ returning
$p$ if $b$ evaluates to true and $q$ otherwise.

\section{A Simple Hybrid Programming Language}\label{sec:language}

Let us build a simple hybrid programming language to illustrate some
of our challenges and results. Intuitively, this language adds
differential equation constructs to the standard imperative features,
namely assignments, sequencing, and conditional branching.  It was
first presented in \cite[Chapter 3]{neves18} and we will use this
paper's results to extend it with a notion of iteration. We start by
recalling the definition of the hybrid monad~\cite{NevesBarbosaEtAl16}
here denoted by~$\BBH_0$, as a candidate semantic domain for this
language. In the following sections, we will extend~$\BBH_0$ in order
to obtain additional facilities for interpreting progressive and
hybrid iteration.
\begin{definition}[\cite{NevesBarbosaEtAl16}]\label{defn:H-monad}
The monad $\BBH_0$ on $\Set$ is defined in the following manner.
\begin{itemize}
  \item The set $H_0X$ has as elements the pairs $(d,e)$ 
    with $d\in\Rze$ and $e:\Rz\to X$ a function satisfying
    the \emph{flattening condition}: for every $x\geq d$,
    $e(x) = e(d)$. We call the elements of $(d,e)$ \emph{duration} and
    \emph{evolution}, respectively, and use the subscripts $\dr$ and~$\ev$
    to access the corresponding fields, i.e.\ given
    $f=(d,e)\in H_0X$, we mean $f_\dr$ and $f_\ev$ to denote $d$
    and~$e$ respectively.  This convention extends to Kleisli
    morphisms as follows: given $f:X\to H_0Y$,
    $f_\dr(x) = (f(x))_\dr$, $f_\ev(x) = (f(x))_\ev$.
  \item The unit is defined by $\eta(x) = (0,\const x)$, where
    $\const x$ denotes the constant trajectory on $x$;
  \item For every Kleisli morphism $f:X\to H_0Y$ and every value $(d,e)\in H_0X$,
\begin{align*} 
  (f^\klstar (d,e))_\dr = \ite{d + f_\dr(e^d)}{d\in\Rz}{\infty}\qquad\quad
  (f^\klstar (d,e))_\ev^t = \ite{f^0_\ev(e^t)}{t \leq d}{f_\ev^{t-d}(e^{d})}
\end{align*}
(recall that for a pair $f(x) = (d,e)$, according to our conventions,
$(f_\ev(x))^0$ refers to $(f_\ev (x))(0)$; here we additionally
simplify $(f_\ev (x))^0$ to $f^0_\ev (x)$ for the sake of
readability).
\end{itemize}
\end{definition}


\noindent
We now fix a finite set of real-valued variables
$X = \{x_1, \dots, x_n\}$ and denote by $\pAt(X)$ the set of atomic
programs given by the grammar,
\begin{align*}
  & \varphi \ni (x_1 := t, \dots, x_n := t) \mid
    (\dot{x}_1 = t, \dots, \dot{x}_n = t \> \& \> r) \mid
    (\dot{x}_1 = t, \dots, \dot{x}_n = t \> \& \> \psi),  \\
  &
  t \; \ni r \mid r \cdot x \mid t + t, \hspace{0.5cm}
  \psi \ni t \leq t \mid t \geq t \mid
    \psi \wedge \psi \mid \psi \vee \psi
\end{align*}
where $x \in X$ and $r \in \Rz$. The next step is to construct an
interpretation map,
\begin{align}\label{eq:atoms}
 \lsem \blank \rsem : \pAt(X) \to \End_{\BBH_0}(\Reals^n)
\end{align}
that sends atomic programs $\prog[a]$ to endomorphisms
$\lsem \prog[a] \rsem : \Reals^n \to H_0(\Reals^n)$ in the Kleisli
category of $\BBH_0$. This map extends to terms and predicates as
$\lsem t \rsem{(v_1,\dots,v_n)} \in \Reals^n$ and
$\lsem {\psi} \rsem \subseteq \Reals^n$ in the standard way by structural induction.
We interpret each assignment
$(x_1 := t, \dots, x_n := t)$ as the map,
\begin{flalign*}
  (v_1,\dots,v_n) \mapsto \eta_{\Reals^n} \left (\lsem {t_1} \rsem{(v_1,\dots,v_n)}
    ,\dots ,\lsem {t_n} \rsem{(v_1,\dots,v_n)} \right )
\end{flalign*}
\noindent
Recall that linear systems of ordinary differential equations
$\dot{x}_1 = t, \dots, \dot{x}_n = t$ always have unique solutions
$\phi: \Reals^n \to (\Reals^n)^{\Rz}$ \cite{perko13}. We use this
property to interpret each program
$(\dot{x}_1 = t, \dots, \dot{x}_n = t \> \& \> r)$ as the
respective solution $\Reals^n \to (\Reals^n)^{\Rz}$ but restricted to
$\Reals^n \to (\Reals^n)^{[0,r]}$. In order to interpret programs of
the type $(\dot{x}_1 = t, \dots, \dot{x}_n = t \> \& \> \psi)$ we can
call on the following
result. 
%
\begin{theorem}[\cite{dahlqvist18}]
Consider a program
$(\dot{x}_1 = t, \dots, \dot{x}_n = t \> \& \> \psi)$, the solution
$\phi : \Reals^n \times \Rz \to \Reals^n$ of the system
$\dot{x}_1 = t, \dots, \dot{x}_n = t$, and a valuation
$(v_1,\dots,v_n) \in \Reals^n$. If there exists a time instant
$r \in \Rz$ such that
$\phi(v_1,\dots,v_n,r) \in \lsem {\psi} \rsem$ then there exists a
\emph{smallest} time instant that also satisfies this condition.
\end{theorem}
\noindent
Using this theorem, we interpret each program
$(\dot{x}_1 = t, \dots, \dot{x}_n = t \> \& \> \psi)$ as the
function 
defined by,
\begin{flalign*}
  (v_1,\dots,v_n) \mapsto (d, \phi(v_1,\dots,v_n, -))
\end{flalign*}
\noindent
where $d$ is the smallest time instant that intersects
$\lsem{\psi}\rsem$ if
$(\img\ \phi(v_1,\dots,v_n,-)) \cap \lsem{\psi} \rsem \neq \emptyset$
and $\infty$ otherwise. This final step provides the desired
interpretation map of atomic programs~\eqref{eq:atoms}.

\noindent
We can now systematically build the hybrid programming language using
standard algebraic results, as observed in
\cite{dahlqvist18,neves18}. The set $\End_{\BBH_0}(\Reals^n)$ of endomorphisms
$\Reals^n \to H_0(\Reals^n)$ together with Kleisli composition
$\bullet$ and the unit $\eta : \Id \to H_0$ form a monoid
$(\End_{\BBH_0}(\Reals^n), \bullet, \eta)$. Therefore, the free monoidal
extension of
$\lsem \blank \rsem :\pAt(X) \to (\End_{\BBH_0}(\Reals^n), \bullet, \eta)$ is
well-defined and induces a semantics for program terms,
\begin{ceqn}
\[
  \prog[p
  = a\in At(X) \mid skip\mid p\sComp p]
\]
  \end{ceqn}
\begin{example}
Let us consider some programs written in this language.
\begin{enumerate}
\item We can have classic, discrete assignments, such as
  $\mathsf{x := x + 1}$ or $\mathsf{x:= 2 \cdot x}$, and their
  sequential composition.
\item We can also write a $\mathsf{wait(r)}$ call, frequently used in
  the context of embedded systems for making the system halt its
  execution during $\mathsf{r}$ time units. This is achieved with the
  program $(\dot{x}_1 = 0, \dots, \dot{x}_n = 0 \> \& \> \mathsf{r})$.
\item It is also possible to consider oscillators using histeresis
  \cite{goebel04}, in particular via the sequential composition
  $(\dot{\mathtt{x}} = 1 \> \& \> 1) \sComp (\dot{\mathtt{x}} = -1 \>
  \& \> 1)$.
\item The bouncing ball system that was examined in the introduction
  is another program of this language.
\end{enumerate}
\end{example}

\noindent
We next extend our language with if-then-else clauses.  This can be
achieved in the following manner. Denote by $B$ the free Boolean
algebra generated by the expressions $t = t$ and $t < t$.  Each
$b \in B$ induces an obvious predicate map $\sem{b} : \Reals^n \to 2$.

Any $b$ induces a binary function
$+_b : \End_{\BBH_0}(\Reals^n) \times \End_{\BBH_0}(\Reals^n) \to \End_{\BBH_0}(\Reals^n)$
defined as follows: $(f +_b g)(x) = \ite{f(x)}{b(x)}{g(x)}$.
This allows us to freely extend the
interpretation map,
\begin{flalign*}
\lsem \blank \rsem :\pAt(X) \to (\End_{\BBH_0}(\Reals^n), \bullet, \eta,
(+)_{b \in B})
\end{flalign*}
into a hybrid programming language with if-then-else clauses
$\prog[p +_{b \in B} p]$.

\begin{example} Let us consider some programs of this language
  with control decision features.
  \begin{enumerate}
  \item Aside from while-loops, our language carries the basic features of
    classic programs with discrete assignments, sequential
    composition, and if-then-else constructs.
  \item The (simplistic) cruise controller,
  $(\dot{\mathsf{v}} = 1 \> \& \> 1)  +_{\prog[v] \leq 120}
  (\dot{\mathsf{v}} = -1 \> \& \> 1)$
  discussed in the introduction is also a program of this
  language.
  \end{enumerate}
\end{example}

\noindent
To be able to address more complex behaviours we need some
means for forming iterative computations, such as \emph{while-loops}
\begin{ceqn}
\begin{flalign}\label{eq:while}
  \prog[ while ] \> \prog[ b ]  \> \prog[ \{ p \} ]
\end{flalign}
\end{ceqn}
This poses the main challenge of our present work, which is to give a
semantics of such constructs w.r.t.\ to a suitably designed
\emph{hybrid monad}. As a starting point, we refer
to~\cite{NevesBarbosaEtAl16,neves18} where $\BBH_0$ and an iteration
operator $(\argument)^\hash:\Hom(X,H_0X)\to\Hom(X,H_0X)$, which we
call \emph{basic iteration}, were introduced. One limitation of
this approach can already be read from the type profile:
$(\argument)^\hash$ can only interpret non-terminating loops, of the
form $\prog[ while ] \> \> \prog[ true ] \> \prog[ \{ p \} ]$.  The
semantics of $(\argument)^\hash$ in $\BBH_0$ is given by virtue of
metric spaces and Cauchy sequences, making difficult to identify the
corresponding domain of definiteness.  Here we take a different avenue
of introducing an Elgot iteration~\eqref{eq:elgot-iter}, for which, as
we shall see, the monad $\BBH_0$ must be modified. We then show
(in Section~\ref{sec:iter}) that basic iteration can be recovered,
albeit with a semantics subtly different from the one via $\BBH_0$.

\section{Guarded Monads and Elgot Iteration}\label{sec:guarded}
\begin{figure}[t!]
\begin{gather*}
\kern-1ex\textbf{(trv)}\quad\frac{f:X\to TY}{~(T\inj_1)\comp f:X\to_{\inj_2} T(Y+Z)~}\qquad
\textbf{(sum)}\quad\frac{~f:X\to_\sigma TZ\qquad~g:Y\to_\sigma TZ}
{~[f,g]:X+Y\to_\sigma TZ}\quad\\[3ex]
\kern-1ex\textbf{(cmp)}\quad\frac{~f:X\to_{\inj_2} T(Y+Z)\qquad g:Y\to_{\sigma} TV\qquad h:Z\to TV~}{[g,h]^\klstar\comp f:X\to_{\sigma} TV} \qquad
\end{gather*}
%
\caption{Axioms of abstract guardedness.}
\label{fig:guard}
\end{figure}

We proceed to give the necessary definitions related to guardedness for
monads~\cite{GoncharovSchroderEtAl17}.  A monad~$\BBT$ (on $\Set$) is
\emph{(abstractly) guarded} if it is equipped with a notion of
guardedness, which is a relation between Kleisli morphisms
$f:X\to TY$ and injections $\sigma:Y'\ito Y$ closed under the rules in
Fig\,\ref{fig:guard} where $f:X\to_\sigma Y$ denotes the fact that $f$
and $\sigma$ are in the relation in question. In the sequel, we also
write $f:X\to_i TY$ for $f:X\to_{\inj_i} TY$. More generally, we use
the notation $f:X\to_{p,q,\ldots} TY$ to indicate guardedness in the
union of injections $\inj_p, \inj_q,\ldots$ where $p,q,\ldots$ are
sequences over $\{1,2\}$ identifying the corresponding coproduct
summand in~$Y$. For example, we write $f:X\to_{12,2} T((Y+Z)+Z)$ to
mean that $f$ is $[\inj_1 \comp \inj_2,\inj_2]$-guarded.


\begin{definition}[Guarded Elgot monads]\label{def:g-elgot}
A monad $\BBT$ is a \emph{guarded Elgot monad} if it is equipped with a \emph{guarded
iteration operator},
\begin{align*}
  (f:X\to_2 T(Y+X))\mto (f^\istar:X\to TY)
\end{align*}
satisfying the following laws:
\begin{itemize}
  \item\emph{fixpoint law}: $f^\istar = [\eta,f^\istar]^\klstar\comp f$;
  \item\emph{naturality:} $g^{\klstar} f^{\istar} = ([(T\inl) \comp g, \eta\inr]^{\klstar} \comp f)^{\istar}$ for $f:X\to_2 T(Y+X)$, $g : Y \to TZ$;
  \item\emph{codiagonal:} $(T[\id,\inr] \comp f)^{\istar} = f^{\istar\istar}$ for  $f : X \to_{12,2} T((Y + X) + X)$;
  \item\emph{uniformity:} $f \comp h = T(\id+ h) \comp g$ implies
	$f^{\istar} \comp h = g^{\istar}$ for $f: X \to_2 T(Y + X)$, $g: Z \to_2 T(Y + Z)$ and
	$h: Z \to X$.
\end{itemize}
We drop the adjective `guarded' for guarded Elgot monads for which guardedness is total,
i.e.\ $f:X\to_\sigma TY$ for any $f:X\to TY$ and $\sigma$.
%
%
\end{definition}
\begin{figure}[t!]
  \tikzset{
    font=\tiny,
    nonterminal/.style={
      rectangle,
      minimum size=6mm,
      very thick,
      draw=orange!50!black!50,         
      fill=orange!50!white,
      font=\itshape
    },
    terminal/.style={
      scale=.5,
      circle,
      inner sep=0pt,
      thin,draw=black!50,
      top color=white,bottom color=black!20,
      font=\ttfamily
    },
    iterated/.style={
      fill=green!20,
      thick,
      draw=green!50
    },
    natural/.style={
      circle,
      minimum size=4mm,
      inner sep=2pt,
      thin,draw=black!50,
      top color=white,bottom color=black!20,
      font=\ttfamily},
    skip loop/.style={to path={-- ++(0,#1) -| (\tikztotarget)}},
    o/.style={
      shorten >=#1,
      decoration={
        markings,
        mark={
          at position 1
          with {
            \fill[black!55] circle [radius=#1];
          }
        }
      },
      postaction=decorate
    },
    o/.default=2.5pt,
    p/.style={
      shorten <=#1,
      decoration={
        markings,
        mark={
          at position 0
          with {
            \fill[black!55] circle [radius=#1];
          }
        }
      },
      postaction=decorate
    },
    p/.default=2.5pt
  }

  {
    \tikzset{nonterminal/.append style={text height=1.5ex,text depth=.25ex}}
    \tikzset{natural/.append style={text height=1.5ex,text depth=.25ex}}
  }
  \captionsetup[subfigure]{labelformat=empty,justification=justified,singlelinecheck=false}
  \pgfdeclarelayer{background}
  \pgfdeclarelayer{foreground}
  \pgfsetlayers{background,main,foreground}
    \begin{subfigure}{\textwidth}
    \centering
    \caption{Fixpoint:}
    \vspace{-2ex}
    \raisebox{-.5\height}{
    \begin{tikzpicture}[
      point/.style={coordinate},>=stealth',thick,draw=black!50,
      tip/.style={->,shorten >=0.007pt},every join/.style={rounded corners},
      hv path/.style={to path={-| (\tikztotarget)}},
      vh path/.style={to path={|- (\tikztotarget)}},
      text height=1.5ex,text depth=.25ex 
      ]
      \node [nonterminal] (f) {$f$};
      \draw [<-] (f.west) -- +(-1,0) node [midway,above] {$X$};
      \path [o,<-,draw] ($(f.west)+(-0.5,0)$) -- +(0,-0.8) -| ($(f.east)+(0.5,-0.15)$) node [pos=0.8,right] {$X$} -- +(-0.5,0);
      \draw [->] (f.east)++(0,0.15) -- +(1,0) node [midway,above] {$Y$};
      
      \begin{pgfonlayer}{background}
        \draw [iterated] ($(f.north west)+(-0.25,0.25)$) rectangle ($(f.south east)+(0.25,-0.25)$);
      \end{pgfonlayer}
    \end{tikzpicture}
    }
    ~~=~~
    \raisebox{-.5\height}{
    \begin{tikzpicture}[
      point/.style={coordinate},>=stealth',thick,draw=black!50,
      tip/.style={->,shorten >=0.007pt},every join/.style={rounded corners},
      hv path/.style={to path={-| (\tikztotarget)}},
      vh path/.style={to path={|- (\tikztotarget)}},
      text height=1.5ex,text depth=.25ex 
      ]
      \node [nonterminal] (f) {$f$};
      \node [nonterminal] (f2) at ($(f.east)+(1.5,-0.15)$) {$f$};
      \draw [<-] (f.west) -- +(-1,0) node [midway,above] {$X$};
      \draw [p,->] (f.east)++(0,-0.15) -- (f2.west) node [pos=0.25,below] {$X$};
      \path [o,<-,draw] (f2.west)++(-0.5,0) -- ++(0,-0.8) -| ($(f2.east)+(0.5,-0.15)$) node [near end,right] {$X$} -- +(-0.5,0);
      \draw [->] (f.east)++(0,0.15) -- node [midway,above] {$Y$} ++(0.5,0) -- ++(0,0.5) -- ++(3,0);
      \draw [->] (f2.east)++(0,0.15) -- ++(0.5,0) -- node [midway,right] {$Y$} ++(0,0.65);

      \begin{pgfonlayer}{background}
        \draw [iterated] ($(f2.north west)+(-0.25,0.25)$) rectangle ($(f2.south east)+(0.25,-0.25)$);
      \end{pgfonlayer}
    \end{tikzpicture}
    }
  \end{subfigure}
  \par
  \begin{subfigure}{\textwidth}
    \centering
    \caption{Naturality:}
    \raisebox{-.5\height}{
    \begin{tikzpicture}[
      point/.style={coordinate},>=stealth',thick,draw=black!50,
      tip/.style={->,shorten >=0.007pt},every join/.style={rounded corners},
      hv path/.style={to path={-| (\tikztotarget)}},
      vh path/.style={to path={|- (\tikztotarget)}},
      text height=1.5ex,text depth=.25ex 
      ]
      \node [nonterminal] (f) {$f$};
      \node [nonterminal] (g) at ($(f.east)+(1.5,0.15)$) {$g$};
      \draw [<-] (f.west) -- +(-1,0) node [midway,above] {$X$};
      \path [o,<-,draw] (f.west)++(-0.5,0) -- ++(0,-0.8) -| ($(f.east)+(0.5,-0.15)$) node [near end,right] {$X$}  -- ++(-0.5,0);
      \draw [->] ($(f.east)+(0,0.15)$) -- (g) node [midway,above] {$Y$};
      \draw [->] (g.east) -- +(1,0) node [midway,above] {$Z$};
      \begin{pgfonlayer}{background}
        \draw [iterated] ($(f.north west)+(-0.25,0.25)$) rectangle ($(f.south east)+(0.25,-0.25)$);
      \end{pgfonlayer}
    \end{tikzpicture}
    }
    ~~=~~
    \raisebox{-.5\height}{
    \begin{tikzpicture}[
      point/.style={coordinate},>=stealth',thick,draw=black!50,
      tip/.style={->,shorten >=0.007pt},every join/.style={rounded corners},
      hv path/.style={to path={-| (\tikztotarget)}},
      vh path/.style={to path={|- (\tikztotarget)}},
      text height=1.5ex,text depth=.25ex 
      ]
      \node [nonterminal] (f) {$f$};
      \node [nonterminal] (g) at ($(f.east)+(1.5,0.15)$) {$g$};
      \draw [<-] (f.west) -- +(-1,0) node [midway,above] {$X$};
      \path [o,<-,draw] (f.west)++(-0.5,0) 
        -- ++(0,-0.75) 
        -| ($(g.east)+(0.5,-0.5)$) node [near end,right] {$X$} 
        -| ($(f.east)+(0.6,-0.15)$) 
        -- ++(-0.6,0);
      \draw [->] ($(f.east)+(0,0.15)$) -- (g) node [midway,above] {$Y$};
      \draw [->] (g.east) -- +(1,0) node [midway,above] {$Z$};
      \begin{pgfonlayer}{background}
        \draw [iterated] ($(f.north west)+(-0.25,0.25)$) rectangle ($(g.south east)+(0.25,-0.35)$);
      \end{pgfonlayer}
    \end{tikzpicture}
    }
  \end{subfigure}
  \par\medskip
  \begin{subfigure}{\textwidth}
    \centering
    \vspace{1ex}
    \caption{Codiagonal:}
    \vspace{-1ex}
    \raisebox{-.5\height}{
    \begin{tikzpicture}[
      point/.style={coordinate},>=stealth',thick,draw=black!50,
      tip/.style={->,shorten >=0.007pt},every join/.style={rounded corners},
      hv path/.style={to path={-| (\tikztotarget)}},
      vh path/.style={to path={|- (\tikztotarget)}},
      text height=1.5ex,text depth=.25ex 
      ]
      \node [nonterminal,minimum height=1.2cm] (f) {$g$};
      \draw [<-] (f.west) -- +(-1,0) node [midway,above] {$X$};
      \draw [->] (f.east)++(0,0.4) -- ++(1.5,0) node [pos=0.3,above] {$Y$};
      \draw [p,->] (f.east) -- ++(1,0) --node [midway,right] {$X$}  ++(0,-1.1) -| ($(f.west)+(-0.5,0)$);
      \draw [p,->] (f.east)++(0,-0.4) -- node [midway,below] {$X$} ++(0.5,0) -- ++(0,0.4);
      \begin{pgfonlayer}{background}
        \draw [iterated] ($(f.north west)+(-0.25,0.25)$) rectangle ($(f.south east)+(0.75,-0.25)$);
      \end{pgfonlayer}
    \end{tikzpicture}
    }
    ~~=~~
    \raisebox{-.5\height}{
    \begin{tikzpicture}[
      point/.style={coordinate},>=stealth',thick,draw=black!50,
      tip/.style={->,shorten >=0.007pt},every join/.style={rounded corners},
      hv path/.style={to path={-| (\tikztotarget)}},
      vh path/.style={to path={|- (\tikztotarget)}},
      text height=1.5ex,text depth=.25ex 
      ]
      \node [nonterminal,minimum height=1cm] (f) {$g$};
      \draw [<-] (f.west) -- +(-1.6,0) node [pos=0.7,above] {$X$};
      \draw [->] (f.east)++(0,0.3) -- ++(1.5,0) node [pos=0.3,above] {$Y$};
      \draw [p,->] (f.east) -- ++(1.1,0) -- node [midway,right] {$X$} ++(0,-1.35) -| ($(f.west)+(-0.95,0)$);
      \draw [p,->] (f.east)++(0,-0.3) -- ++(0.5,0) -- node [midway,right] {$X$} ++(0,-0.65) -| ($(f.west)+(-0.5,0)$);
      \begin{pgfonlayer}{background}
        \draw [iterated] ($(f.north west)+(-0.75,0.35)$) rectangle ($(f.south east)+(0.85,-0.6)$);
        \draw [iterated,fill=green!55] ($(f.north west)+(-0.25,0.2)$) rectangle ($(f.south east)+(0.25,-0.25)$);
      \end{pgfonlayer}
    \end{tikzpicture}
    }
  \end{subfigure}
  \par
  \begin{subfigure}{\textwidth}
    \caption{Uniformity:}
    \centering
\begin{tabular}{rcl}
   \raisebox{-.5\height}{
    \begin{tikzpicture}[
      point/.style={coordinate},>=stealth',thick,draw=black!50,
      tip/.style={->,shorten >=0.007pt},every join/.style={rounded corners},
      hv path/.style={to path={-| (\tikztotarget)}},
      vh path/.style={to path={|- (\tikztotarget)}},
      text height=1.5ex,text depth=.25ex 
      ]
      \node [nonterminal,fill=blue!20,draw=blue!50] (h) {$h$};
      \node [nonterminal] (f) at ($(h.east)+(1.5,0)$) {$f$};
      \draw [<-] (h.west) -- +(-1,0) node [midway,above] {$Z$};
      \draw [->] (h.east) -- (f.west) node [midway,above] {$X$};
      \draw [->] (f.east)++(0,0.15) -- +(1,0) node [midway,above] {$Y$};
      \draw [p,->] (f.east)++(0,-0.15) -- +(1,0) node [midway,below] {$X$};
    \end{tikzpicture}
    }
    &$~~=~~$&
    \raisebox{-.5\height}{
    \begin{tikzpicture}[
      point/.style={coordinate},>=stealth',thick,draw=black!50,
      tip/.style={->,shorten >=0.007pt},every join/.style={rounded corners},
      hv path/.style={to path={-| (\tikztotarget)}},
      vh path/.style={to path={|- (\tikztotarget)}},
      text height=1.5ex,text depth=.25ex 
      ]
      \node [nonterminal] (f) {$g$};
      \node [nonterminal,fill=blue!20,draw=blue!50] (h) at ($(f.east)+(1.5,-0.15)$) {$h$};
      \draw [<-] (f.west) -- +(-1,0) node [midway,above] {$Z$};
      \draw [p,->] (f.east)++(0,-0.15) -- (h.west) node [midway,below] {$Z$};
      \draw [->] (f.east)++(0,0.15) -- ++(0.65,0) -- ++(0,0.4) node [midway,left] {$Y$} -- ++(2.35,0);
      \draw [->] (h.east) -- +(1,0) node [midway,above] {$X$};
    \end{tikzpicture}
    }\\
    &\Large{$\Downarrow$}&\\
    \raisebox{-.5\height}{
    \begin{tikzpicture}[
      point/.style={coordinate},>=stealth',thick,draw=black!50,
      tip/.style={->,shorten >=0.007pt},every join/.style={rounded corners},
      hv path/.style={to path={-| (\tikztotarget)}},
      vh path/.style={to path={|- (\tikztotarget)}},
      text height=1.5ex,text depth=.25ex 
      ]
      \node [nonterminal,fill=blue!20,draw=blue!50] (h) {$h$};
      \node [nonterminal] (f) at ($(h.east)+(1.5,0)$) {$f$};
      \draw [<-] (h.west) -- +(-1,0) node [midway,above] {$Z$};
      \draw [->] (h.east) -- (f.west) node [midway,above] {$X$};
      \draw [->] (f.east)++(0,0.15) -- +(1,0) node [midway,above] {$Y$};
      \path [o,<-,draw] (f.west)++(-0.5,0) -- ++(0,-0.8) -|
      ($(f.east)+(0.5,-0.15)$) node [pos=0.8,right] {$X$} -- ++(-0.5,0);
      \begin{pgfonlayer}{background}
        \draw [iterated] ($(f.north west)+(-0.25,0.25)$) rectangle ($(f.south east)+(0.25,-0.25)$);
      \end{pgfonlayer}
    \end{tikzpicture}
    }
    &$~~=~~$&
    \raisebox{-.5\height}{
    \begin{tikzpicture}[
      point/.style={coordinate},>=stealth',thick,draw=black!50,
      tip/.style={->,shorten >=0.007pt},every join/.style={rounded corners},
      hv path/.style={to path={-| (\tikztotarget)}},
      vh path/.style={to path={|- (\tikztotarget)}},
      text height=1.5ex,text depth=.25ex 
      ]
      \node [nonterminal] (f) {$g$};
      \draw [<-] (f.west) -- +(-1,0) node [midway,above] {$Z$};
      \path [o,<-,draw] ($(f.west)+(-0.5,0)$) -- +(0,-0.8) -| ($(f.east)+(0.5,-0.15)$) node [pos=0.8,right] {$Z$} -- +(-0.5,0);
      \draw [->] (f.east)++(0,0.15) -- +(1,0) node [midway,above] {$Y$};
      
      \begin{pgfonlayer}{background}
        \draw [iterated] ($(f.north west)+(-0.25,0.25)$) rectangle ($(f.south east)+(0.25,-0.25)$);
      \end{pgfonlayer}
    \end{tikzpicture}
    }
  \end{tabular}
  \end{subfigure}
\vspace{2ex}
\caption{Axioms of guarded iteration.}
\label{fig:ax}
\end{figure}
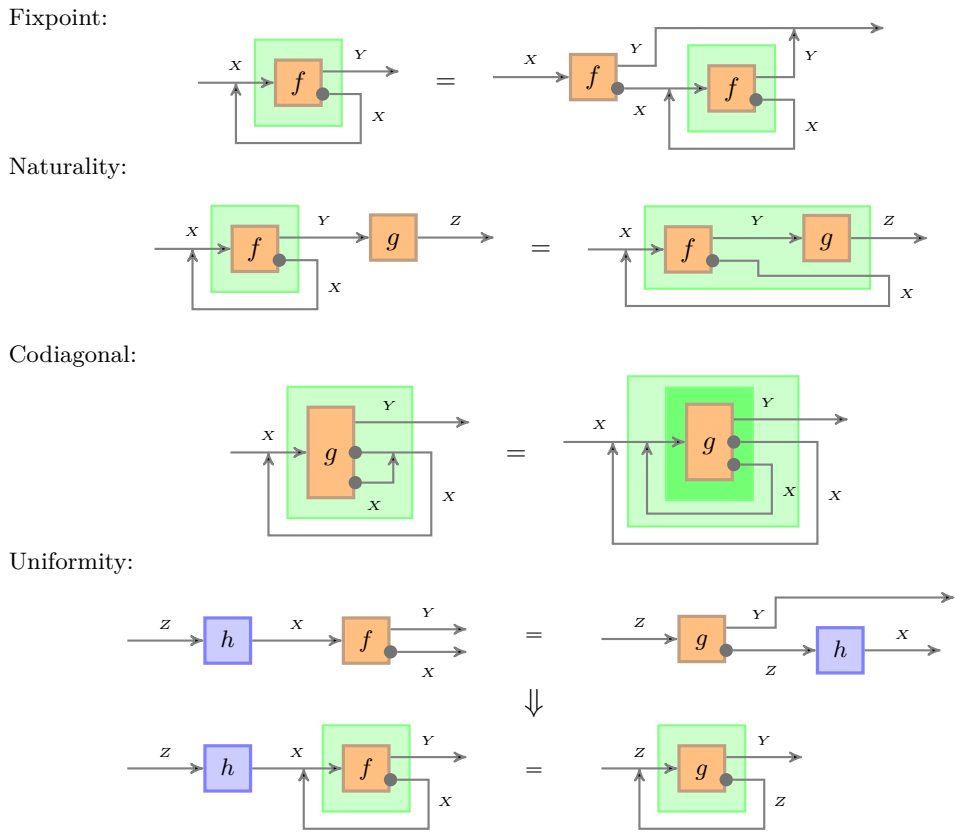


\noindent
The notion of guarded monad is a common generalisation of various
cases occurring in practice. Every monad can be equipped with a least
notion of guardedness, called~\emph{vacuous guardedness} and defined
as follows: $f:X\to_2 T(Y+Z)$ iff $f$ factors through
$T\inl:T Y\to T(Y+Z)$. Every vacuously guarded monad is guarded Elgot,
for every fixpoint $f^\istar$ unfolds precisely once
\cite{GoncharovSchroderEtAl17}. On the other hand, the greatest notion
of guardedness is \emph{total guardedness} and is defined as follows:
$f:X\to_2 T(Y+Z)$ for every $f:X\to T(Y+Z)$.  This addresses
\emph{total iteration} operators on $\BBT$ (e.g.\ for $\BBT$ being
Elgot), whose existence depends on special properties of~$\BBT$, such
as being enriched over complete partial orders. Motivating examples,
however, are those properly between these two extreme situations,
e.g.\ \emph{completely iterative monads}~\cite{Milius05} for which the
notion of guardedness is defined via monad modules and the iteration
operator is partial, but uniquely satisfies the fixpoint law.

\begin{example}\label{expl:monad}
We illustrate the above concepts with the following simplistic examples.
\begin{enumerate}
\item The powerset monad $\PSet$ is Elgot, with the iteration operator sending
  $f:X\to\PSet(Y+X)$ to $f^\istar:X\to\PSet Y$ calculated as
  the least solution of the fixpoint law
  $f^\istar = [\eta,f^\istar]^\klstar f$.
\item An example of partial guarded iteration can be obtained from the
  previous clause by replacing $\PSet$ with the \emph{non-empty
    powerset monad}~$\PSet_{\mplus}$. The total iteration operator
  from the previous clause does not restrict to a total iteration
  operator on this monad, because empty sets can arise from solving
  systems not involving empty sets, e.g.\
  $\eta \comp \inr:1\to\PSet_{\mplus}(1+1)$ would not have a solution
  in this sense. However, it is easy to see that the iteration
  operator from the previous clause restricts to a guarded one for
  $\PSet$ with the notion of guardedness defined as follows:
  $f:X\to_2\PSet_{\mplus}(Y+X)$ iff for every element $x \in X$,
  $f(x)$ contains at least one element from $Y$.
\end{enumerate}
\end{example}
\noindent
The axioms of guarded Elgot monads are given in
Fig\,\ref{fig:ax} in an intuitive pictorial form. The shaded boxes
indicate the scopes of the corresponding iteration loops and bullets
attached to output wires express the corresponding guardedness
predicates. As shown in~\cite{GoncharovSchroderEtAl17}, other standard
principles such as \emph{dinaturality} and the \emph{Beki\'{c} law}
follow from this axiomatisation.

%
%
%

\section{A Fistful of Hybrid Monads}\label{sec:fist}
According to Moggi~\cite{Moggi91}, Kleisli morphisms can be viewed as
generalised functions carrying a computational effect, e.g.\
non-determinism, process algebra actions, or their combination.  In
this context, hybrid computations can be seen as computations extended
in time.
\begin{wrapfigure}[8]{r}[0pt]{.42\textwidth}
\vspace{0pt}
  \centering
\begin{tikzcd}[column sep = normal,row sep = large]
&
	\BBH_0\BBM
    \ar[dr,hookleftarrow,shift right=-.75ex,"\upsilon"]
    \ar[dr,->>,shift left=-.75ex,"\rho"'] &
\\
\BBH_{\mplus}
    \ar[ur,hookrightarrow, "\iota"]
    \ar[rr,hookrightarrow, "\rho\iota", dashed]     & &
\BBH
\end{tikzcd}
\vspace{-0pt}
  \caption{Connecting $\BBH_0\BBM$, $\BBH_\mplus$ and~$\BBH$}
  \label{fig:connection}
\vspace{-0pt}
\end{wrapfigure}

By definition, the pairs $(d,e)\in H_0X$ fall into two classes:
\emph{closed trajectories} with $d\neq\infty$ and \emph{open
  trajectories} with $d=\infty$. Due to the flattening condition (see
Definition~\ref{defn:H-monad}), closed trajectories are completely
characterized by their restrictions to $[0,d]$.  We proceed by
extending $\BBH_0$ to a larger monad that brings open trajectories
over arbitrary intervals $[0,d)$ with $d>0$ into play, and call the
resulting monad $\BBH_\mplus$. It is instrumental in our study to cope
with open trajectories, as in the presence of Zeno behaviour,
iteration might produce open trajectories $[0,d) \to X$ that we cannot
sensibly extend into $[0,d] \to X$ without assuming some structure on
$X$ \cite{hofner_phdthesis,HofnerMoller11,nakamura05}.  Furthermore,
we introduce a variant of~$\BBH_{\mplus}$, which we call $\BBH$ and
which extends the facilities of $\BBH_{\mplus}$ even further by
including the \emph{empty trajectory} $[0,0)\to X$ which will be used
to accommodate \emph{non-progressive divergent} computation (see
Remark~\ref{rem:onH}). As detailed in the sequel, the mere addition of
the empty trajectory does not exactly fit the bill -- it yields `gaps'
in trajectories, which makes no sense under the assumption that
computations cannot recover from divergence.  To fix this, $\BBH$ will
forbid the extension of computations over time after a divergence
occurs. The notation for~$\BBH$ and $\BBH_{\mplus}$ is chosen to be
suggestive, and is a reminiscent of $\PSet$ and~$\PSet_\mplus$ for the
powerset and the non-empty powerset monads as in
Example~\ref{expl:monad}.  Indeed, the analogy goes further, as in the
next section we show that $\BBH_\mplus$ supports guarded (progressive)
iteration,~$\BBH$ supports total iteration, and the former is a
restriction of the latter.


%
In order to develop $\BBH_\mplus$, we first introduce a partial
version of $\BBH_0$ that will greatly facilitate obtaining some of our
results. Essentially, this partial version amounts to the combination
of $\BBH_0$ with the \emph{maybe monad} $\BBM$. Recall that
$MX = X + 1$, that the unit of $\BBM$ is given by the left coproduct
injection $\inl:X\to X+1$, and that the Kleisli lifting sends
$f:X\to Y+1$ to $[f,\inr]:X+1\to Y+1$. We conventionally identify
Kleisli morphisms $X\to MY$ with partial functions from $X$ to $Y$ and
thus write $f(x)\dar$ to indicate that $f(x)$ is defined on $x$, i.e.\
$f(x) \neq\inr\star$.  Let
$\dom(f) = {\{x\in X\mid f(x)\dar\}}\subseteq X$ and let us denote
by~$\bot$ both $\inr\star\in X+1$ and the totally undefined function
$\const\bot$. Finally, we write $f(x)\uar$ as a shorthand notation to
$f(x)=\bot$. For the sake of readability, we will sometimes write the
composition of partial functions $(g + \id) f$ simply as $g f$.
We will also need the following result.
\begin{proposition}\label{prop:m-distr}
  Every monad $\BBT=(T,\eta,(\argument)^\klstar)$ induces a monad $\BBT\BBM$
  whose functor is defined by $X\mto TMX$, the unit by
  $\eta\inl:X\to TMY$, and the Kleisli lifting by
  $[f,\eta\inr]^\klstar:TMX\to TMY$ for every
  $f:X\to TMY$.
\end{proposition}
\begin{proof}
  This is a consequence of the standard fact that every monad
  distributes over the maybe monad~\cite{ghani02}.
\end{proof}
\begin{definition}\label{defn:H-plus}
Let $\BBH_0\BBM$ be the monad identified in Proposition~\ref{prop:m-distr} with $\BBT=\BBH_0$.
Then let~$H_{\mplus}$ be the subfunctor of $H_0 M$ that is defined by,
\begin{align}\label{eq:constraint1}
(d,e)\in H_{\mplus} X \text{\quad iff\quad} e\neq\bot \text{\quad and\quad} e^t\dar \text{~ for all ~} t\in [0,d).
\end{align}
This yields a monad $\BBH_{\mplus}$, by restricting the monad
structure of $\BBH_0 \BBM$.  Explicitly, $\eta(x) = (0,\const x)$ and for
every $f:X\to H_{\mplus}Y$ and every $(d,e)\in H_{\mplus}X$,
\begin{flalign*}
\quad&  (f^\klstar (d,e))_\dr = d&&
  (f^\klstar (d,e))_\ev^t = \ite{f^0_\ev(e^t)}{t<d}{\bot} & \text{(if~~$e^d\uar$)}\\
\quad& (f^\klstar (d,e))_\dr = d + f_\dr(e^d)&&
  (f^\klstar (d,e))_\ev^t = \ite{f^0_\ev(e^t)}{t< d}{f_\ev^{t-d}(e^d)} & \text{(if~~$e^d\dar$)}
\end{flalign*}
\end{definition}
Note that the set $H_{\mplus} X$ consists precisely of elements
$(d,e)$ for which either $\dom(e) = \Rz$ or $\dom(e) = [0,d)$ and
$d>0$. Of course, we need to verify that Definition~\ref{defn:H-plus}
correctly introduces a monad.
\begin{proof}
  We only need to show that for every $f:X\to H_{\mplus}Y$,
  $(d,e)\in H_{\mplus} X$ implies that
  $f^\klstar(d,e)\in H_{\mplus}Y$. Let $t\in [0,(f^\klstar(d,e))_\dr)$
  and proceed by case distinction:
  \begin{citemize}
  \item $e^d\uar$. Then $(f^\klstar(d,e))_\dr= d$ and $t<d$. Since
    $(d,e) \in H_{\mplus} X$, the condition $e^t \dar$ holds and
    consequently $f^0_\ev(e^t) \dar$. Then since
    $(f^\klstar (d,e))_\ev^t = f^0_\ev(e^t) $, we have
    $(f^\klstar (d,e))_\ev^t \dar$ which proves our claim.

\item $e^d\dar$. Then $(f^\klstar (d,e))_\ev^t\dar$ iff either
  $t < d$ and $f^0_\ev(e^t)\dar$ or $t> d$ and
  $f_\ev^{t-d}(e^d)\dar$. In the former case we are done in the same
  way as in the previous clause. In the latter case, note that
  $t-d<(f^\klstar(d,e))_\dr - d = f_\dr(e^d)$, which by assumption
  implies that $f_\ev^{t-d}(e^d)\dar$.\qed
\end{citemize}
\noqed\end{proof}
The condition $e\neq\bot$ in~\eqref{eq:constraint1} is essential for
the construction above, for otherwise we cannot ensure that
computations with totally undefined trajectories are compatible with
Kleisli liftings, as detailed in Remark~\ref{rem:n-strict} below. Such
computations can be seen as representing unproductive or
non-progressive divergence since they do not progress in time. They
are required for the semantics of programs like
\begin{ceqn}
  \begin{align*}
  \mathsf{while}  \> \mathsf{true} \> \{
  {\mathtt{x}} := \mathtt{x} + 1 \}
\end{align*}
\end{ceqn}
We therefore need to extend $\BBH_{\mplus}$ to a larger monad $\BBH$
in which such divergent computations exist. Technically, this will amount
to quotienting the monad $\BBH_0\BBM$ in a suitable manner.
\begin{definition}\label{defn:H}
Let $\BBH_0\BBM$ be the monad identified in Proposition~\ref{prop:m-distr} with $\BBT=\BBH_0$
and let~$H$ be the subfunctor of $H_0 M$ formed as follows:
\begin{align}\label{eq:constraint2}
(d,e)\in H X &&\text{iff} && e \text{~ is total \quad or\quad} d=\infty\text{\quad and \quad} \dom e \text{~ is downward closed.}
\end{align}
The total trajectories included in $H$ must be understood in precisely
the same way as in $H_\mplus$ while the remaining trajectories fall
into two classes:
\begin{itemize}
  \item $(\infty,e)$ with $\dom e=[0,d)$ -- these correspond to the trajectories
$(d,e)$ of $\BBH_{\mplus}$, unless $d=0$, in which case we obtain a counterpart of
the empty trajectory not included in $\BBH_\mplus$;
  \item $(\infty,e)$ with $\dom e=[0,d]$ -- these trajectories behave analogously,
but have no counterparts in $\BBH_{\mplus}$.
\end{itemize}
Both these cases are meant to model divergent
behaviours, with the moment of divergence occurring either at the time
instant $d$ in the first case, or immediately after $d$ in the second
case.

Let $\upsilon$ be the inclusion of $H$ into $H_0M$ and let
$\rho:H_0M\to H$ be the natural transformation whose components are
defined by,
\begin{align*}
(\rho_X(d,e))_\dr = \ite{d}{\dom e = \Rz}{\infty}, \hspace{.5cm} (\rho_X(d,e))_\ev^t
= \ite{e^t}{t\leq d_\star}{(\ite{e^{d_\star}}{\dom e = \Rz}{\bot})}
\end{align*}
where $d_\star=\sup\{t<d\mid [0,t)\subseteq\dom e \}$. It is easy to
see that $\rho$ is a right inverse of $\upsilon$.

We extend $\BBH$ to a monad by defining $x\mto \rho(\eta(x))$ to be
the unit and the Kleisli lifting the map sending $f:X\to HY$ to
$\rho(\upsilon f)^\klstar\upsilon$. Explicitly, the monad structure on
$H$ is as follows:
$\eta(x) = (0,\const x)$ and for every $f:X\to HY$, and every
$(d,e)\in HX$, assuming that $D=\bigcup\,\bigl\{[0,t]\subseteq\dom f_\ev^0\comp e\mid [0,t]\subseteq\dom e\bigr\}$,
\begin{flalign*}
~~  (f^\klstar (d,e))_\dr =&\; d + f_\dr(e^d),& (f^\klstar (d,e))_\ev^t =&\; \ite{f^0_\ev(e^t)}{t \leq d}{f_\ev^{t-d}(e^d)}&
(\text{if~$D=\Rz$}) \\
~~  (f^\klstar (d,e))_\dr =&\; \infty, & (f^\klstar (d,e))_\ev^t =&\; \ite{f^0_\ev(e^t)}{t \in D}{\bot} & (\text{otherwise})
\end{flalign*}
\end{definition}
Like in the case of $\BBH_\mplus$
, we need
to verify that $\BBH$ is a monad  (see Appendix~\ref{a:H} for details).
\begin{remark}\label{rem:n-strict}
  As indicated above, $\BBH$ is a quotient of $\BBH_0\BBM$ and not a submonad,
  specifically~$\upsilon$ is not a monad morphism. Indeed, given $f:\Rz\to
  H\Rz$ such that $f(0) = (\infty,\bot)$ and $f(t) = (1,\const 1)$ for $t>0$,
  computing $f^\klstar (\infty,\id)$ w.r.t.\ $\BBH_0\BBM$ yields $(f^\klstar
  (\infty,\id))_\ev^0 = \bot$ and $(f^\klstar (\infty,\id))_\ev^t = 1$ for $t>0$,
  which does not belong to $H\Rz$.
\end{remark}
In summary, the monads $\BBH_0\BBM$, $\BBH_\mplus$, $\BBH$ are
connected as depicted in Fig\,\ref{fig:connection}. Here, $\iota$ and
$\rho$ are monad morphisms, and the induced composite morphism
$\rho\iota:\BBH_\mplus\to\BBH$ is pointwise injective.

\section{Progressive Iteration and Hybrid Iteration}\label{sec:iter}
We start off by equipping the monad $\BBH_{\mplus}$ from the previous section with a suitable notion of guardedness.
\begin{definition}[Progressiveness]
A Kleisli morphism $(d,e):X\to H_\mplus(Y+Z)$ is \emph{progressive} in $Z$ (in $Y$) if
$e^0:X\to Y+Z$ factors through $\inl$ (respectively, $\inr$).
\end{definition}
Given $(d,e):X\to H_\mplus(Y+X)$, progressiveness in $X$ means
precisely that $e^0=\inl u:X\to Y+X$ for a suitable $u:X\to Y$, which
is intuitively the candidate for $(d,e)^\istar_\ev$ at $0$. In other
words, progressiveness rules out the situations in which the iteration
operator needs to handle compositions of zero-length trajectories.

\begin{remark}\label{rem:onH}
A simple example of a morphism $(d,e):X\to H_\mplus(Y+X)$ not
progressive in~$X$ is obtained by taking $X=\{0,1\}$, $Y=\emptyset$,
$d=\const 0$ and $e^0 = \inr\oname{swap}$ where $\oname{swap}$
interchanges the elements of $\{0,1\}$.  In attempts of defining
$(d,e)^\istar$ we would witness oscillation between~$0$ and $1$
happening at time $0$, i.e.\ not progressing over time, which is
precisely the reason why there is no candidate semantic for
$(d,e)^\istar$ in this case.
\end{remark}
\begin{lemma}\label{lem:guard}
$\BBH_\mplus$ is a guarded monad with $f:X\to_2H_{\mplus}(Y+Z)$ iff $f$ is
progressive in $Z$.
\end{lemma}

\noindent
Instead of directly equipping $\BBH_{\mplus}$ with progressive iteration,
we take the following route: we enrich the monad $\BBH_0\BBM$ over
complete partial orders and devise a total iteration operator for it using
the standard least-fixpoint argument. Then we restrict iteration from
$\BBH_0\BBM$ to~$\BBH_\mplus$ via $\iota$ and to $\BBH$ via~$\upsilon$ (see
Fig\,\ref{fig:connection}). The latter part is tricky, because~$\upsilon$
is not a monad morphism (Remark~\ref{rem:n-strict}), and thus we will call
on the machinery of~\emph{iteration-congruent retractions}, developed
in~\cite{GoncharovSchroderEtAl17}, to derive a (total) Elgot iteration on
$\BBH$.

Consider the following order on $H_0M X$: for
$(d,e), (d_\star,e_\star)\in H_0M X$,
$(d,e)\appr(d_\star,e_\star)$ if
\begin{align*}
d\leq d_\star,~	e\leq e_\star \text{\quad and\quad} d\in\Rz,~e^d\dar \text{\quad imply\quad} d=d_\star
\end{align*}
%
where evolutions are compared as partial maps, i.e.\ $e\leq e_\star$ reads as
$\dom(e)\subseteq\dom(e_\star)$ and $e^t = e_\star^t$ for all $t\in\dom(e)$.

This order extends to the hom-sets $\Hom(X,H_0MY)$ pointwise.
\begin{theorem}\label{thm:H0M}
The following properties hold.
\begin{enumerate}
\item\label{item:H0M1} Every set $H_0MX$ is an $\omega$-complete partial order under
  $\appr$ with $(0,\bot)$ as the bottom element;
  \item\label{item:H0M2} Kleisli composition is monotone and continuous w.r.t.\ $\appr$ on both sides;
  \item\label{item:H0M3} Kleisli composition is right-strict, i.e.\ for every $f:X\to H_0MY$, $f^\klstar (0,\bot)= (0,\bot)$.
\end{enumerate}
\end{theorem}

\noindent
Note that Kleisli composition is not left strict, e.g.\ $(0,\bot)^\klstar (1,\bot) = (1,\bot)\neq (0,\bot)$.
Using Theorem~\ref{thm:H0M} and a previous result~\cite[Theorem 5.8]{GoncharovRauchEtAl15}, we immediately obtain
\begin{corollary}\label{cor:H0M-iter}
$\BBH_0\BBM$ possesses a total iteration operator $(\argument)^\iistar$ obtained as
a least solution of equation $f^\iistar = [\eta,f^\iistar]^\klstar f$. This makes
$H_0MX$ into an Elgot monad.
Explicitly, $f^\iistar$ is calculated via the \emph{Kleene fixpoint
  theorem} as follows. For $f:X\to H_0M{(Y+X)}$, let
$f^{\brks{0}} = (0,\bot)$ and
$f^{\brks{i+1}} = [\eta,f^{\brks{i}}]^\klstar f$. This yields an
$\omega$-chain 
\begin{displaymath}
  f^{\brks{0}}\appr f^{\brks{1}}\appr\cdots
\end{displaymath}
and $f^\iistar = \bigjoin_i f^{\brks{i}}$. 
\end{corollary}
We readily obtain a \emph{progressive iteration} on $\BBH_\mplus$ by restriction
via $\iota$ (see~Fig\,\ref{fig:connection}).
\begin{corollary}\label{cor:deriv}
$\BBH_\mplus$ possesses a guarded iteration operator $(\argument)^\pstar$ by restriction from
$\BBH_0\BBM$ with guardedness being progressiveness and $f^\pstar$ being the least
solution of equation $f^\pstar = [\eta,f^\pstar]^\klstar f$.
\end{corollary}
We proceed to obtain an iteration operator for~$\BBH$.
Remarkably, we
cannot use the technique of restricting the iteration operator
from $\BBH_0\BBM$ to $\BBH$, we applied in the case of $\BBH_{\mplus}$, even though $H$ embeds into $H_0M$ -- the
following example illustrates the issue.
\begin{example}\label{exp:iter-H}
  Let $f=(\const{1},e): \Rz\to H_0M(\Rz+\Rz)$ with $e(x) = \const{\inr 0}$ if $x =
  0$ and $e(x) = \const{\inl 1}$ otherwise. Even though $f$ factors through the
  inclusion $\upsilon:H\to H_0M$, the result of calculating $f^\iistar(0)$ is a
  trajectory $(1,e_\star)$ with $\dom e = (0,\infty)$, which is \emph{not} down-closed
  and therefore $(1,e_\star)$ is not in $\BBH$.
\end{example}
Example~\ref{exp:iter-H} indicates that the restriction of the canonical complete 
partial order from $H_0MX$ to $HX$ is not complete and therefore we cannot use 
it to show that $\BBH$ is Elgot. We can nevertheless obtain the following
\begin{theorem}\label{thm:retract} 
Let $\rho:\BBH_0\BBM\to\BBH$ and $\upsilon:\BBH\to\BBH_0\BBM$ be the pair of natural
transformations from Definition~\ref{defn:H}. Then for every $f:X\to H_0M(Y+X)$,
$\rho f^\iistar = \rho (\upsilon\rho f)^\iistar$.
\end{theorem}
In the terminology of~\cite{GoncharovSchroderEtAl17}, Theorem~\ref{thm:retract}
states that the pair $(\rho,\upsilon)$ is an \emph{iteration-congruent retraction}.
Therefore, per~\cite[Theorem 21]{GoncharovSchroderEtAl17}, $\BBH$ inherits a total Elgot
iteration from $\BBH_0\BBM$.
\begin{corollary}
  $\BBH$ is an Elgot monad with the iteration operator
  $(\argument)^\istar$ defined as follows: for every $f:X\to H(Y+X)$,
  $f^\istar = \rho (\upsilon f)^\iistar$ assuming that
  $(\argument)^\iistar$ is the iteration operator on $\BBH_0 \BBM$.
\end{corollary}
\begin{corollary}\label{cor:restr}
The progressive iteration operator $(\argument)^\pstar$ of $\BBH_{\mplus}$ is
the restriction of the total iteration operator $(\argument)^\istar$ of
$\BBH$ along $\rho\iota:H_\mplus\to H$ (as in Fig\,\ref{fig:connection}), i.e.\
for every $f:X\to_2 H_{\mplus}(Y+X)$, $\rho\comp \iota f^\pstar = (\rho\comp \iota f)^\istar$.
\end{corollary}
\begin{proof}
Let $(\argument)^\iistar$ be the iteration operator of $\BBH_0\BBM$. Then, by definition,
\begin{flalign*}
&& (\rho\comp \iota f)^\istar = \rho(\upsilon\rho\comp \iota f)^\iistar = \rho(\iota f)^\iistar= \rho\comp \iota f^\pstar.&&\qed
\end{flalign*}
\noqed
\end{proof}
Using the fact that the iteration operator for $\BBH$ satisfies the codiagonal
law (see Definition~\ref{def:g-elgot}), we factor the former through progressive
iteration as follows.
\begin{theorem}[Decomposition Theorem]\label{thm:decomp} Given $f:X\to H(Y+X)$, let
$\hat f:X\to_{12} H((Y+X)+X)$ be defined as follows:
\begin{flalign*}
\quad  \hat f_\dr(x) = f_\dr(x) \qquad \hat f_\ev^0(x) = (\inl+\id)(f_\ev^0(x)) \quad \hat f_\ev^t(x) = \inl f_\ev^t(x) && (x\in X, t>0)
\end{flalign*}
Then $f^\istar = (\hat f^\istar)^\pstar$.
\end{theorem}
\begin{proof}
Note that $f = H[\id,\inr]\hat f$. Hence, by the codiagonal law: $f^\istar = (H[\id,\inr]\hat f)^\istar = \hat f^{\istar\istar}$,
and the latter is $(f^\istar)^\pstar$ per Corollary~\ref{cor:restr}, as $f^\istar$
happens to be progressive in the second argument.
\end{proof}
Theorem~\ref{thm:decomp} presents the iteration of $\BBH$ as a nested
combination of progressive iteration and what can be called
\emph{singular iteration}, as it is precisely the restriction of $(\argument)^\istar$
responsible for iterating computations of zero duration.
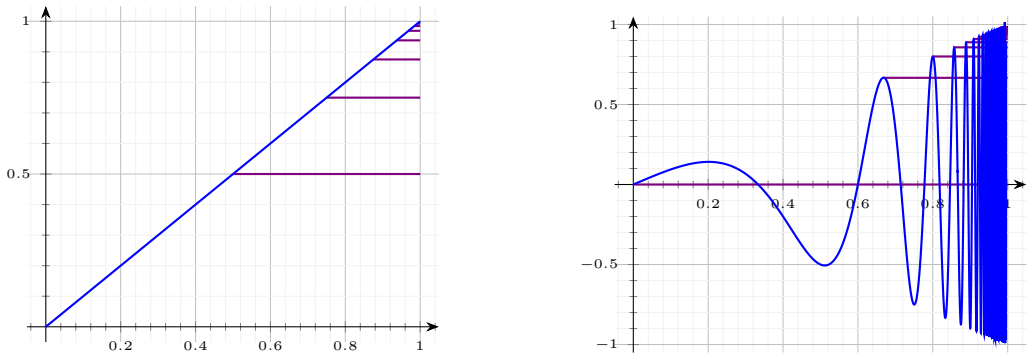
\begin{figure}[t!]
\centering
\begin{tikzpicture}
\begin{axis}[xmin=0, xmax=1, ymin=0, ymax=1,width=.5\textwidth]

\foreach \x [evaluate=\x as \y using {1-2^(-\x)}] in {1, ..., 8} {
  \edef\temp{\noexpand\addplot[violet, thick, smooth, samples at={\y,1}] {\y};}\temp
}

\addplot[blue, thick, smooth,
  samples at={0,1}] {x};
\end{axis}
\end{tikzpicture}
\hspace{1.5cm}
\begin{tikzpicture}
\begin{axis}[xmin=0, xmax=1, ymin=-1, ymax=1,width=.5\textwidth]

\foreach \x [evaluate=\x as \y using {2 * \x / (1 + 2 * \x)}] in {0, ..., 30} {
  \edef\temp{\noexpand\addplot[violet, thick, smooth, samples at={\y,1}] {\y};}\temp
}

\addplot[blue, thick, smooth,
  samples at={0,.005,...,
              .6,.601,.602,...,
              .713,.714,.715,...,
              .776,.777,.778,...,
              .818,.8181,.8185,...,
              .845,.846,.847,...,
              .865,.866,.8661,...,
              .995
}] {x * cos(pi * x * deg(1/(1-x))};
\end{axis}
\end{tikzpicture}
\caption{Examples of (un-)definiteness of basic iteration.}
\label{fig:sin-example}
\end{figure}
Finally, we recover basic iteration, discussed in
Section~\ref{sec:language}, on $\BBH_\mplus$ (and hence on $\BBH$) by
turning a morphism $X\to H_\mplus X$ into a progressive one
$X\to H_\mplus(X+X)$.
\begin{definition}[Basic Iteration]
  We define \emph{basic iteration}
  $(d,e)^\hash:X\to H_\mplus X$ to be
\begin{displaymath}
  \big ((d,\lambda x. \, \lambda t.\,\ite{\inl e^0(x)}{t=0}{\inr e^t(x)} ):X
  \to_2  H_\mplus (X+X) \big )^\pstar:X\to H_\mplus X.
\end{displaymath}
\end{definition}
\begin{example}\label{expl:h-inter}
  We illustrate our design decisions behind $\BBH_\mplus$ (and $\BBH$)
  with the following two examples of Zeno behaviour, computing
  $f^\hash = (d^\hash,e^\hash)$ for specific morphisms
  $f:X\to H_\mplus X$.
\begin{enumerate}
\item Let $f=(d,e):[0,1]\to H_\mplus[0,1]$ be defined as follows: for
  every $x\in [0,1]$, $d(x) = (1 - x)/2$ and $e^t(x) = x + t$ for
  $t\in [0,(1 - x)/2]$. It is easy to see that
  $d^\hash(0) = 1/2 + 1/4 + \ldots = 1$, however, by definition,
  $(e^\hash(0))^1\uar$. This is indeed a prototypical example of Zeno
  behaviour (specifically, this is precisely Zeno's \emph{``Dichotomy''
    paradox} analyzed by Aristotle~\cite[Physics, 
231a–241b]{Aristotle08}): Given a distance of total length~$1$ to be covered,
suppose some portion $x<1$ of it has been covered already. Then the remaining
distance has the length $1-x$. As originally argued by Zeno, in order to cover
this distance, one has to pass the middle, i.e.\ walk the initial interval
of length $(1-x)/2$ and our function~$f$ precisely captures the dynamics
of this motion. The resulting evolution $e^\hash(0)$ together with the
corresponding approximations are depicted on the left of
Fig\,\ref{fig:sin-example}. In this formalization, the traveller can not reach
the end of the track, but only because we designed $(\argument)^\hash$ to be
so. We could also justifiably define $(e^\hash(0))^1$ to be $1$, for this is what
$(e^\hash(0))^t$ tends to as $t$ tends to $1$. This is indeed the case of the
approach from~{\cite{NevesBarbosaEtAl16,neves18}} developed for the original monad $\BBH_0$.
 \item It is easy to obtain an example of an open trajectory produced by Zeno iteration
that cannot be continuously extended to a closed one by adapting a standard example of
\emph{essentially discontinuous} function
from analysis: let e.g.\ $u^t: [0,1)\to [0,1)$ be as follows:
\begin{flalign*}
  \qquad u^t(x) =&\; (t+x)\cos\left(\frac{\pi t}{(1-x) (1-x-t)}\right) & (t \in [0,1-x))\\
  \qquad u^t(x) =&\; 1 & (t\in [1-x,1))
\end{flalign*}
The graph of $u(0)$ is depicted on the right of Fig\,\ref{fig:sin-example} where
one can clearly see the discontinuity at $t=1$. It is easy to verify that $(1,u)\in H[0,1]$
is obtained by applying basic iteration to
$f=(d,e):[0,1)\to H_\mplus[0,1)$ given as follows:
\begin{flalign*}
  \quad d(x)   =\;& \frac{2(1-x)^2}{3-2x}\qquad&
  e^t(x) =\;& (t+x)\cos\left(\frac{\pi t}{(1-x) (1-x-t)}\right)\qquad & (t \in [0,d(x)))\\[1ex]
  &&e^t(x) =\;& d(x) + x & (t\in [d(x),1))
\end{flalign*}
\end{enumerate}
\end{example}
Even though we carried our developments in the category of sets, we
designed $\BBH_\mplus$ and $\BBH$ keeping in touch with a topological
intuition. The following instructive example shows that the iteration
operators developed in the previous section cannot be readily
transferred to the category of topological spaces and continuous maps,
for reasons of \emph{instability}: small changes in the definition of
a given system may cause drastic changes in its behaviour. In
particular, even if a morphism $(d,e) : X \to H_0 X$ is continuous
(for the topology described in \cite{NevesBarbosaEtAl16}) the duration
component $d^\hash : X \to [0,\infty]$ of
$(d^\hash,e^\hash) = (d,e)^\hash$ need not be continuous.
\begin{example}[Hilbert Cube]\label{exp:hilbert}
Let $X=[0,1]^\omega$ be the \emph{Hilbert cube}, i.e.\ the topological
product of $\omega$ copies of $[0,1]$ and let $\oname{hd}:X\to [0,1]$ and
$\oname{tl}:X\to X$ be the obvious projections realizing the isomorphism
$[0,1]^\omega\cong [0,1]\times [0,1]^\omega$. Let $f=(\oname{hd},e):X\to H_\mplus X$
with $e:X\to X^{\Rz}$ be defined as follows:
\begin{itemize}
  \item $e^t(x) = x$ if $\oname{hd}(x) = 0$, $t\in\Rz$;
  \item $e^t(x) = \bigl((\oname{hd}(x) - t)\cdot x + t \cdot \oname{tl}(x)\bigr)/\oname{hd}(x)$ if $0<\oname{hd}(x)$ and $t < \oname{hd}(x)$;
  \item $e^t(x) = \oname{tl}(x)$ if $\oname{hd}(x) > 0$ and $t\geq\oname{hd}(x)$.
\end{itemize}
In the second clause we use a \emph{convex combination} of $x$ and
$\oname{tl}(x)$ as vectors of $X$ seen as a vector space (indeed, even
a Hilbert space) over the reals.  It can now be checked that the
cumulative duration $d^\hash$ in $(d^\hash,e^\hash) = (d,e)^\hash$ is
not continuous. To see why, note that $d^\hash(x)$ is the (possibly
infinite) sum of the components of $x$ from left to right up to the
first zero element, and therefore each $U=(d^\hash)^\mone([0,a))$
contains all such vectors $x\in [0,1]^\omega$ for which this sum is
properly smaller than $a$. Then recall that a basic open set of $[0,1]^\omega$ must be
a \emph{finite} intersection of sets of the form $\pi_i^{\mone}(V)$,
$V\subseteq[0,1]$ open, $i \in \mathbb{N}$. Therefore, if $U$ was open the
definition of the product topology on $[0,1]$ would imply that for
every vector $x$ in $U$ there exists a position such that by altering
the components of $x$ arbitrarily after this position, the result
would still belong to $U$. This is obviously not true for $U$, because
by replacing the elements of any infinite vector from $[0,1]^\omega$
after any position with $1$, would give a vector summing to
infinity.
\end{example}

\section{Bringing While-loops Into The Scene}\label{sec:sem}

In Section~\ref{sec:language}, we started building a simple hybrid
programming language. We sketched a monad-based semantics for the
expected programs constructs, except the while-loops. Here we extend
it by taking $\BBH$, which is a supermonad of $\BBH_0$, as the
underlying monad and interpret while-loops~\eqref{eq:while} via the
iteration operator of $\BBH$.

Recall that $\prog[b]$ is an element of the free Boolean
algebra generated by the expressions $t = t$ and $t < t$, and that
there exists a predicate map $\prog[b] : \Reals^n \to 2$. Now for each
$\prog[b] : \Reals^n \to 2$ and $f : \Reals^n \to H(\Reals^n)$
denote the function,
\[
%
  \Big ( \Reals^n \xto{~\dist \comp\pv{\id,\prog[b]}~}
         \Reals^n + \Reals^n \xto{~[\eta\inl,\, (H \inr) \comp f ]~}
         H(\Reals^n + \Reals^n) \xto{~m~} H(\Reals^n + \Reals^n)  \Big )^\istar
\]
by $w(\prog[b],f)$ where $\dist : X \times 2 \to X + X$ is the obvious distributivity
transformation, and
$m(d,e) = (d, e')$ with
$e'(t) = \ite{\inl(x)}{(\inr (x) = e(t) \text{ and } t <
  d)}{e(t)}$. Intuitively, the function $m$ makes the last point of the
trajectory be the only one that is evaluated by the test condition
of the while-loop. Then, we define
$\lsem \prog[ while ] \> \prog[b] \> \{ \prog[p] \} \rsem =
w(\prog[b],\lsem \prog[p] \rsem)$ and this gives a hybrid programming
language,
\begin{ceqn}
\[
  \prog[p = a\in At(X) \mid skip\mid p\sComp p \mid p +_{b} p \mid
  while \> b \> \{ p \}]
\]
\end{ceqn}
\noindent with while-loops.
\begin{example}
  Let us consider some programs written in this language.
\begin{enumerate}
\item We start again with a classic program, in this case
  $\prog[ while ] \> \prog[ true ] \> \prog[ \{ x:= x + 1 \} ]$.  It
  yields the empty trajectory $\bot$.
\item Another example of a classic program is,
\begin{ceqn}
\begin{flalign*}
 \prog[ while  \> x \leq 10  \>  \{ x:= x + 1 \sComp wait(1) \} ]
\end{flalign*}
\end{ceqn}
If for example the initial value is $0$ the program takes eleven time
units to terminate.
\item Let us consider now the program
  $\prog[ while \> x \geq 1 \> \{ \> ( \dot{x} = - 1 \> \& \> 1 ) \>
  \} ]$.  If the initial value is $0$ the program outputs the
  trajectory with duration $0$ and constant on $0$, since it never
  enters in the loop. If we start e.g. with $3$ as initial value then
  the program inside the while-loop will be executed precisely
  three times, continuously decreasing $\mathsf{x}$ over time.
\item In contrast to classic programming languages, here infinite
  while-loops need not be undefined. The cruise controller discussed
  in the introduction,
\begin{ceqn}
  \begin{align*}
  \mathsf{while}  \> \mathsf{true} \> \{
  (\dot{\mathtt{v}} = 1 \> \& \> 1)  +_{\prog[v] \leq 120}
  (\dot{\mathtt{v}} = -1 \> \& \> 1) \}
\end{align*}
\end{ceqn}
\noindent
is a prime example of this.
\item Finally, the bouncing ball,
\begin{ceqn}
  $(\mathsf{p := 1, v := 0}) ;  ( \prog[ while \> true \> \{ b \}] )$
which has Zeno behaviour, outputs a trajectory describing the ball's
movement over the time interval $[0,d)$ where $d$ is the instant of
time at which the ball stops.
\end{ceqn}

\end{enumerate}
\end{example}

\section{Conclusions and Further Work}\label{sec:concl}
We developed a semantics for hybrid iteration by bringing together two
abstraction devices introduced recently: guarded Elgot
iteration~\cite{GoncharovSchroderEtAl17} and the hybrid
monad~{\cite{NevesBarbosaEtAl16,neves18}}. Our analysis reveals that,
on the one hand, the abstract notion of guardedness can be interpreted
as a suitable form of progressiveness of hybrid trajectories, and on the
other hand, the original hybrid monad
from~{\cite{NevesBarbosaEtAl16,neves18}} needs to be completed for the
sake of a smooth treatment of iteration, specifically, iteration
producing Zeno behaviour.  In our study we rely on Zeno behaviour
examples as important test cases helping to design the requisite
feasible abstractions. As another kind of guidance, we rely on Elgot's
notion of iteration~\cite{Elgot75} and the corresponding laws of
iteration theories~\cite{BloomEsik93}. In addition to the new hybrid
monad $\BBH_\mplus$ equipped with (partial) progressive iteration, we
introduced a larger monad $\BBH$ with total hybrid iteration extending
the progressive one. In showing the iteration laws we heavily relied on
the previously developed machinery for unifying guarded and unguarded
iteration~\cite{GoncharovRauchEtAl15,GoncharovSchroderEtAl17}.  We
illustrated the developed semantic foundations by introducing a simple
language for hybrid iteration with while-loops interpreted over the
Kleisli category of $\BBH$.

We regard our present work as a stepping stone for further
developments in various directions. After formalizing hybrid
computations via (guarded) Elgot monads, one obtains access to further
results involving (guarded) Elgot monads, e.g.\ it might be
interesting to explore the results of applying the \emph{generalized
  coalgebraic resumption monad
  transformer}~\cite{GoncharovRauchEtAl15} to $\BBH$ and thus obtain
in a principled way a semantic domain for hybrid processes in the
style of CCS. As shown by Theorem~\ref{thm:decomp}, the iteration of
$\BBH$ is a combination of progressive iteration and `singular
iteration'. An interesting question for further work is if this combination can
be framed as a universal construction.
We also would like to place $\BBH$ in a
category more suitable than $\Set$, but as
Example~\ref{exp:hilbert} suggests, this is expected to be a very
difficult problem.

Every monad on $\Set$ determines a corresponding \emph{Lawvere theory}, whose
presentation in terms of operations and equations is important for reasoning about
the corresponding -- in our case hybrid -- programs. We set as a goal for further research the task of identifying
the underlying Lawvere theories of hybrid monads and integrating them into generic
diagrammatic reasoning in the style of Fig\,\ref{fig:ax}. This should prospectively
connect our work to the line of research by Bonchi, Soboci{\'n}ski, and Zanasi (see e.g\,\cite{Bonchi14,Bonchi17}), who
studied various axiomatizations of PROPs (i.e\,monoidal generalizations of Lawvere
theories) and their diagrammatic languages. For a proper treatment of guarded
iteration (i.e.\ a specific instance of guarded monoidal
trace in the sense of~\cite{GoncharovSchroder18}), one would presumably need to develop
the corresponding notions of \emph{guarded Lawvere theory} and \emph{guarded PROP}.




\clearpage
\bibliography{hybrid}

\clearpage
\appendix
\allowdisplaybreaks

\section{Appendix: Omitted proofs}

\subsection{Proof that $\BBH$ is a monad}\label{a:H}
Note that $\rho$ is a pointwise retraction with $\upsilon$ as a section.
Hence each $HX$ is a quotient of $H_0MX$. We are thus left to show that $\rho$
preserves the monad structure. This is by definition for the unit. For Kleisli
lifting this amounts to the equation $\rho f^\kklstar = (\rho f)^\klstar\rho$,
for every $f:X\to H_0MY$ where we denote by $f^\kklstar$ the Kleisli lifting
of the monad $\BBH_0\BBM$ to distinguish it from the Kleisli lifting of $\BBH$.

Let $(d,e)\in H_0MX$, and let $(d_\star,e_\star) = \rho(d,e)\in HX$,
with
\begin{displaymath}
d_\star=\sup\{t<d\mid [0,t)\subseteq\dom e\},\qquad e_\star^t = \ite{e^t}{t\leq d_\star}{(\ite{e^{d_\star}}{\dom e =\Rz}{\bot})}.
\end{displaymath}
Since $f^\klstar = \rho (\upsilon f)^\kklstar\upsilon$, we need to check that
$\rho f^\kklstar(d,e) = \rho(\upsilon\rho f)^\kklstar (d_\star,e_\star)$, which
we obtain by transitivity from the following equations:
\begin{align}
  \rho f^\kklstar(d,e) =\;& \rho f^\kklstar(d_\star,e_\star),\label{eq:H-mon-1}\\
  \rho f^\kklstar(d_\star,e_\star) =\;&
\rho(\upsilon\rho f)^\kklstar (d_\star,e_\star).\label{eq:H-mon-2}
\end{align}
Let us show~\eqref{eq:H-mon-1} first. If $\dom e = \Rz$, i.e.\ $e$ is
a total function then $d_\star=d$,
$e=e_\star$ and~\eqref{eq:H-mon-1} holds trivially. Otherwise, it
turns into
\begin{displaymath}
  \rho f^\kklstar(d,e) = \rho f^\kklstar(\infty,e_\star)
\end{displaymath}
and the fact that $e$ is not total implies that $e^t\uar$ for some $t\leq d$. Then,
for this $t$, $(f^\kklstar(d,e))_\ev^t\uar$ and thus
$(\rho f^\kklstar(d,e))_\dr=\infty=(\rho f^\kklstar(\infty,e_\star))_\dr$. Let
\begin{align*}
  c=\sup\{t<(f^\kklstar(d,e))_\dr\mid [0,t)\subseteq\dom (f^\kklstar(d,e))_\ev\},
\end{align*}
and, as we have argued,
$c\leq d$. Note that
\begin{flalign*}
&&  c=&\;\sup\{t<(f^\kklstar(d,e))_\dr\mid [0,t)\subseteq\dom (f^\kklstar(d,e))_\ev\}\\
&&   =&\;\sup\{t<d_\star\mid [0,t)\subseteq\dom f_\ev^0\comp e\} &\by{since $e^t\uar$ for some $t\leq d$}\\
&&   =&\;\sup\{t<d_\star\mid [0,t)\subseteq\dom f_\ev^0\comp e_\star\}\\
&&   =&\;\sup\{t<(f^\kklstar(\infty,e_\star))_\dr\mid [0,t)\subseteq\dom (f^\kklstar(\infty,e_\star))_\ev\}.&\by{since $(f^\kklstar(\infty,e_\star))_\dr=\infty$}
\end{flalign*}
Now, since by definition $c\leq d$,
\begin{flalign*}
&&(\rho f^\kklstar(d,e))_\ev^t=&\; (f^\kklstar(d,e))_\ev^t=f_\ev^0(e^t)=(f^\kklstar(\infty,e_\star))_\ev^t=(\rho f^\kklstar(\infty,e_\star))_\ev^t &&& \text{if~~} t\in [0,c)\\
&&(\rho f^\kklstar(d,e))_\ev^t=&\; (f^\kklstar(d,e))_\ev^c=f_\ev^0(e^c)=(f^\kklstar(\infty,e_\star))_\ev^c=(\rho f^\kklstar(\infty,e_\star))_\ev^t &&& \text{if~~} t = c
\end{flalign*}
In summary, we obtained $(\rho f^\kklstar(d,e))_\ev=(\rho f^\kklstar(\infty,e_\star))_\ev$, as desired.

We proceed with the proof of~\eqref{eq:H-mon-2}. Equivalently, we replace it with
\begin{align*}
  \rho f^\kklstar(d,e) = \rho(\upsilon\rho f)^\kklstar (d,e),
\end{align*}
where $(d,e)$ falls into one of the following cases: $e$ is a total trajectory or $\dom e \neq \Rz$
with $d=\infty$. In the latter situation, we have $(f^\kklstar(d,e))_\dr = ((\upsilon\rho f)^\kklstar(d,e))_\dr = \infty$,
$(f^\kklstar(d,e))_\ev^t = f^0_\ev(e^t)$ and $((\upsilon\rho f)^\kklstar(d,e))_\ev^t=(\upsilon\rho f)^0_\ev(e^t)= f^0_\ev(e^t)$.
That is,  $f^\kklstar(d,e)$ and $(\upsilon\rho f)^\kklstar (d,e)$ are equal. Hence
they remain equal after applying $\rho$.

Finally, consider the case of total $e$. Again, we make use of the general fact that
\mbox{$(\upsilon\rho f)^0_\ev=f^0_\ev$}. Then, by definition,
\begin{align*}
   (f^\kklstar(d,e))_\dr =&\; d + f_\dr(e^d), &    (f^\kklstar(d,e))_\ev^t =&\; \ite{f^0_\ev(e^t)}{t<d}{f^{t-d}_\ev(e^d)},\\
   ((\upsilon\rho f)^\kklstar(d,e))_\dr =&\; d + (\upsilon\rho f)_\dr(e^d), &    ((\upsilon\rho f)^\kklstar(d,e))_\ev^t =&\; \ite{f^0_\ev(e^t)}{t<d}{(\upsilon\rho f)^{t-d}_\ev(e^d)}.
\end{align*}
Let $c=\sup\{t<d + f_\dr(e^d)\mid [0,t)\subseteq\dom (f^\kklstar(d,e))_\ev\}$. Now, if $f^0_\ev(e^t)\uar$
for some $t<d$ then $c<d$, $c=\sup\{t<d + (\upsilon\rho f)_\dr(e^d)\mid [0,t)\subseteq\dom ((\upsilon\rho f)^\kklstar(d,e))_\ev\}$
and thus for all $t$,
\begin{flalign*}
&&(\rho f^\kklstar(d,e))_\ev^t
=&\; \ite{(f^\kklstar(d,e))_\ev^t}{t\leq c}{\bot}\\
&&=&\; \ite{f^0_\ev(e^t)}{t\leq c}{\bot}&\by{using $c<d$}\\
&&=&\; \ite{((\upsilon\rho f)^\kklstar(d,e))_\ev^t}{t\leq c}{\bot}&\by{using $c<d$}\\
&&=&\;(\rho (\upsilon\rho f)^\kklstar(d,e))_\ev^t,
\end{flalign*}
which yields~\eqref{eq:H-mon-2}. Assume now that $f^0_\ev(e^t)\dar$ for all $t<d$,
which implies $c\geq d$. If $f^{t}_\ev(e^d)\dar$ for all $t$ then $c=\infty$
and~\eqref{eq:H-mon-2} is easy to see. We proceed under the assumption that
$f^{t}_\ev(e^d)\uar$ for some~  $t$, which implies that either
$\dom(f^\kklstar(d,e))_\ev = [0,c]$ or $\dom(f^\kklstar(d,e))_\ev = [0,c)$.
Now, %
\begin{flalign*}
&&  \sup&\{t< (f(e^d))_\dr\mid [0,t)\subseteq\dom (f(e^d))_\ev\}\\
&&=&\; \sup\{t<d + f_\dr(e^d)\mid [0,t)\subseteq\dom (f^\kklstar(d,e))_\ev\} -d&\!\!\by{since $[0,d)\subseteq\dom (f^\kklstar(d,e))_\ev$} \\
&&=&\; c -d,
\end{flalign*}
and therefore
\begin{align}\label{eq:rho-fed}
  (\rho f (e^d))_\ev^t = \ite{f_\ev^t (e^d)}{t\leq c-d}{\bot}
\end{align}
This entails
\begin{flalign*}
~\sup&\;\{t<d + (\upsilon\rho f)_\dr(e^d)\mid [0,t)\subseteq\dom ((\upsilon\rho f)^\kklstar(d,e))_\ev\} \\
=&\;\sup \{t\mid [0,t)\subseteq\dom ((\upsilon\rho f)^\kklstar(d,e))_\ev\} &\!\!\by{since $(\upsilon\rho f)_\dr(e^d)=\infty$}\\
=&\; \sup\{t\mid [0,t)\subseteq\dom (\upsilon\rho f(e^d))_\ev\}+d&\hspace{-5ex}\by{since $[0,d)\subseteq\dom ((\upsilon\rho f)^\kklstar(d,e))_\ev$}\\
=&\; \sup\{t\mid [0,t)\subseteq\dom (\rho f(e^d))_\ev\}+d& \\
=&\; \sup\{t< f_\dr(e^d)\mid [0,t)\subseteq\dom (f(e^d))_\ev\}+d &  \\
=&\; c -d +d\\
=&\; c.
\end{flalign*}
For every $t < d$,
$(\rho f^\kklstar(d,e))_\ev^t = (\rho (\upsilon\rho
f)^\kklstar(d,e))_\ev^t$ as before and we are left to check that this
equality is true also for every $t\geq d$, assuming that
$d\neq\infty$. Note that

\begin{flalign*}
&&(\rho f^\kklstar(d,e))_\ev^t
=&\; \ite{(f^\kklstar(d,e))_\ev^t}{t\leq c}{\bot}\\
&&=&\; \ite{(\ite{f^0_\ev(e^t)}{t<d}{f^{t-d}_\ev(e^d)})}{t\leq c}{\bot}\\
&&=&\; \ite{(\ite{f^0_\ev(e^t)}{t<d}{(\upsilon\rho f)^{t-d}_\ev(e^d)})}{t\leq c}{\bot}&\by{\eqref{eq:rho-fed}}\\
&&=&\; \ite{((\upsilon\rho f)^\kklstar(d,e))_\ev^t}{t\leq c}{\bot}\\
&&=&\; (\rho (\upsilon\rho f)^\kklstar(d,e))_\ev^t,
\end{flalign*}
which finishes the proof of~\eqref{eq:H-mon-2}.
\qed

\subsection{Proof of Lemma~\ref{lem:guard}}
Let us verify the axioms.
\begin{itemize}
  \item\textbf{(trv)} Given $(d,e):X\to H_{\mplus}Y$, then $\bigl( (H_{\mplus}\inl)(d(x), e(x))\bigr)_e^0 = \inl e^0(x)$.
\item\textbf{(cmp)} Suppose, $(d,e):X\to_2 H_{\mplus}(Y+Z)$, $g:Y\to_2 H_{\mplus}(V+W)$, $h:Z\to H_{\mplus}(V+W)$. Then
$([g,h]^\star (d,e))_e^0(x) = [g,h]_e^0(e^0(x)) = g_e^0(p(x)) = \inl q(p(x))$
where $p:X\to Y$ and $q:Y\to V$ exist by assumption.
\item \textbf{(sum)} Let $f:X\to_2 H_{\mplus}(Y+Z)$ and $g:Y\to_2 H_{\mplus}(Y+Z)$. Hence
$f_e^0 = \inl p$ and $g_e^0 = \inl q$ for some $p$ and~$q$. Then $[f,g]_e^0 = \inl [p,q]$.
\qed
\end{itemize}

\subsection{Proof of Theorem~\ref{thm:H0M}}
It follows by routine calculations that $\sqsubseteq$ is a partial
order on sets of the type $H_0 M X$, and that $(0,\bot)$ is the bottom
element with respect to this order.

Next we prove that $\sqsubseteq$ is $\omega$-complete,
specifically that every chain of trajectories
\begin{flalign*}
  (d_1,e_1)\appr (d_2,e_2)\appr\ldots
\end{flalign*}
\noindent
has a least upper bound $(d,e)$ with $d=\sup_i d_i$ and for every $t$, $e^t = e_i^t$
if $e_i^t \dar$ for some~$i$ and $e^t=\bot$ if no such $i$ exists.
First, we show that for every $i$ the inequation
$(d_i,e_i) \sqsubseteq (d,e)$ holds. Note that for every
$i$, $d_i \leq d$ and $e_i \leq e$. Moreover, if for some
index $j$, $d_j \in \Rz$ and $e_j^{d_j} \dar$, then $d_j$
is the largest element in the sequence $d_1 \leq d_2 \leq \dots$, and
therefore $d_j = \sup d_i$. This proves that
$(d_i,e_i) \sqsubseteq (d,e)$  for all $i$. Next we show
that if a trajectory $(d_\klstar,e_\klstar) \in H_0MX$ also satisfies
$(d_i,e_i) \sqsubseteq (d_\klstar,e_\klstar)$ for all $i$ then
$(d,e) \sqsubseteq (d_\klstar,e_\klstar)$. Clearly, $d \leq d_\klstar$
and $e \leq e_\klstar$.  Moreover, if $d \in \Rz$ and $e^d \dar$ then
there exists some index $j$ such that $\sup_i d_i = d_j$, and
since $(d_j,e_j) \sqsubseteq (d_\klstar,e_\klstar)$ we have
$d = d_\klstar$.

Our next step is to show that for every function $f : X \to H_0 M Y$,
$(d_1,e_1) \sqsubseteq (d_2,e_2)$ implies $f^\klstar(d_1,e_1) \sqsubseteq f^\klstar(d_2,e_2)$. We
first verify the goal under the assumption that $d_1 = \infty$ or $e_1^{d_1} \uar$. In either
case we have
\begin{flalign*}
  (f^\klstar(d_1,e_1))_{\dr} = d_1 \leq d_2 \leq (f^\klstar(d_2,e_2))_{\dr}.
\end{flalign*}
The fact that
$\dom f^\klstar (d_1,e_1) \subseteq \dom f^\klstar(d_2,e_2)$ is by the following calculation (here
we use composition of partial maps as juxtaposition, e.g.\ $f^0_\ev {e_1}^t_\ev$, without notice):
\begin{flalign*}
  \dom  f^\klstar (d_1,e_1) & =\left \{t \leq d_1 \mid f^0_\ev e^t_1 \dar \right \} \\
  & \subseteq \left \{t \leq d_1 \mid f^0_\ev e_2^t \dar \right \} \\
 & \subseteq \left \{t \leq d_2 \mid f^0_\ev e_2^t \dar \right \} \\
  &  \subseteq \dom f^\klstar(d_2,e_2).
\end{flalign*}
\noindent
The remaining conditions behind
$(d_1,e_1) \sqsubseteq (d_2,e_2)$ are easy to verify. We proceed to analyse the remaining
case of $d_1 \in \Rz$ and $e_1^{d_1} \dar$, which implies $d_1=d_2$, by definition. This
immediately yields the equality of durations
\begin{flalign*}
  (f^\klstar(d_1,e_1))_\dr = d_1 + (f(e_1^{d_1}))_\dr = d_2 + (f(e_2^{d_2}))_\dr = (f^\klstar(d_2,e_2))_\dr.
\end{flalign*}
In regard to $\dom f^\klstar (d_1,e_1) \subseteq \dom f^\klstar(d_2,e_2)$,
we calculate,
\begin{flalign*}
  \dom f^\klstar (d_1,e_1) & =
  \{t \leq d_1 \mid f^0_\ev e_1^t \dar \} \cup
  \{t + d_1 \in \Rz \mid f^{t}_\ev e_1^{d_1} \dar \} \\
  & \subseteq
  \{ t \leq d_2 \mid f^0_\ev e_2^t \dar \} \cup
  \{t + d_1 \in \Rz \mid f^{t}_\ev e_1^{d_1} \dar \}  \\
  &  = \{ t \leq d_2 \mid f^0_\ev e_2^t \dar \} \cup
  \{t + d_2 \in \Rz \mid f^{t}_\ev e_2^{d_2} \dar \} \\
  & = \dom f^\klstar (d_2,e_2).
\end{flalign*}
The remaining conditions behind
$(d_1,e_1) \sqsubseteq (d_2,e_2)$ are again easy to verify.

Next we prove that $f \sqsubseteq g : X \to H_0 M Y$ implies $f^\klstar (d,e) \sqsubseteq g^\klstar (d,e)$
for any trajectory $(d,e) \in H_0 M X$. The
proof that $(f^\klstar(d,e))_\dr \leq (g^\klstar(d,e))_\dr$ follows almost
directly.
Regarding the domains of $f^\klstar(d,e)$ and $g^\klstar(d,e)$, we just
need to calculate
\begin{flalign*}
  \dom  f^\klstar(d,e) & =  \{t \leq d \mid f^0_\ev {e}^t \dar \} \cup
  \{ t + d \in \Rz \mid f^t_\ev {e}^d  \dar \} \\
  & \subseteq  \{t \leq d \mid g^0_\ev {e}^t \dar \} \cup
  \{ t + d \in \Rz \mid g^t_\ev {e}^d \dar \} \\
 & = \dom g^\klstar(d,e).
\end{flalign*}
The previous reasoning also allows us to conclude straightforwardly
that for every $t \in \dom f^\klstar(d,e)$ the equation
$(f^\klstar(d,e))_\ev^t = (g^\klstar(d,e))_\ev^t $ holds.  Finally, as the
last step in showing $f^\klstar(d,e) \sqsubseteq g^\klstar(d,e)$, we need
to prove that $d + f(e^d)_\dr \in \Rz$ and
$(f^\klstar(d,e))_\ev(d + f(e^d)_\dr) \dar$ imply
$(f^\klstar (d,e))_\dr = (g^\klstar(d,e))_\dr$. So assume the left side
of the implication: it entails that $f(e^d)_\ev (f(e^d)_\dr) \dar$, and since
$f (e^d) \sqsubseteq g(e^d)$ we have $f(e^d)_\dr = g(e^d)_\dr$ which
proves that $(f^\klstar (d,e))_\dr = (g^\klstar(d,e))_\dr$. We are thus
done with the proof of Clause~\ref{item:H0M1} of the theorem.

Lets us show Clause~\ref{item:H0M2}. First, we show the equation
\begin{flalign}\label{eq:kl-cont1}
  f^\klstar\left(\bigjoin\nolimits_i ~(d_i,e_i)\right) = \bigjoin\nolimits_i f^\klstar(d_i,e_i)
\end{flalign}
assuming an $\omega$-chain $(d_1, e_1) \sqsubseteq (d_2,e_2) \sqsubseteq \dots$ We start by showing that the durations in the two sides of the
equation are equal by case distinction: first, we
assume that for all $i \in \omega$ either $d_i = \infty$ or
$e_i^{d_i} \uar$, and calculate,
\begin{flalign*}
  \left (f^\klstar\left (\bigjoin\nolimits_i~(d_i,e_i) \right ) \right )_\dr = \left (\bigjoin\nolimits_i~(d_i,e_i) \right )_\dr
  = \left (\bigjoin\nolimits_i~ f^\klstar(d_i,e_i)  \right )_\dr.
\end{flalign*}
\noindent
Moreover,
\begin{flalign*}
  \dom f^\klstar \left (\bigjoin\nolimits_i~ (d_i,e_i) \right )
  & = \left \{ t \leq \sup_i d_i \mid f_\ev^0 e_k^t \dar \text{ for some } k \in \omega \right \} \\
  & = \bigcup\nolimits_i \left \{t \leq d_i \mid f_\ev^0 e_i^t \dar \right \} \\
  & = \dom \left (\bigjoin\nolimits_i~ f^\klstar(d_i,e_i) \right ).
\end{flalign*}
Equation~\eqref{eq:kl-cont1} now follows immediately. We will now assume the existence of some index $j \in \omega$ such
that $d_j \in \Rz$ and $e_j^{d_j} \dar$. This entails $\sup_i d_i = d_j$, which we use to obtain,
\begin{flalign*}
  \left (f^\klstar \left (\bigjoin\nolimits_i~(d_i,e_i) \right  ) \right )_\dr
  & = \left (f^\klstar \left (\bigjoin\nolimits_i\ (d_j,e_{j+i}) \right ) \right )_\dr \\
  & = d_j + (f(e_j^{d_j}))_\dr \\
  & = \left (\bigjoin\nolimits_i~  f^\klstar (d_j,e_{j+i}) \right )_\dr \\
  & = \left (\bigjoin\nolimits_i~  f^\klstar (d_i,e_i) \right )_\dr.
\end{flalign*}
The equality of the domains of $f^\klstar\bigl (\bigjoin\nolimits_i~ (d_i, e_i)\bigr)$ and
$\bigl (\bigjoin\nolimits_i f^\klstar(d_i,e_i)\bigr)$ is established as follows:
\begin{flalign*}
  \dom f^\klstar \left (\bigjoin\nolimits_i~ (d_i,e_i) \right )
  & = \dom f^\klstar \left (\bigjoin\nolimits_i~ (d_j, e_{j+i}) \right ) \\
  & = \left \{ t \leq d_j \mid f_\ev^0 e_{k}^t \dar \text{ for some }
  k \geq j \right \} \cup\bigl \{ t + d_j \in \Rz \mid f_\ev^t e_j^{d_j} \dar \bigr \} \\
  & = \bigcup\nolimits_{i \geq j} \left \{ t \leq d_i \mid f_\ev^0 e_i^t \dar \right \}
  \cup \big \{ t + d_j \in \Rz \mid f_\ev^t e_j^{d_j} \dar  \big \} \\
  & = \dom \left (\bigjoin\nolimits_i f^\klstar(d_j,e_{j+i}) \right ) \\
  & = \dom \left (\bigjoin\nolimits_i f^\klstar(d_i,e_i) \right ).
\end{flalign*}
The requisite equation~\eqref{eq:kl-cont1} is now immediate. Finally, we show that
\begin{flalign}\label{eq:kl-cont2}
  \left(\bigjoin\nolimits_i f_i\right)^\klstar (d,e) = \bigjoin\nolimits_i f_i^\klstar (d,e),
\end{flalign}
for any family of functions $f_i:X \to H_0 M Y$ forming a chain $f_1 \sqsubseteq f_2 \sqsubseteq \ldots$
We proceed again by case distinction. First assume that $d = \infty$ or $e^d \uar$, which
immediately implies
\begin{flalign*}
  \left ( \left (\bigjoin\nolimits_i f_i \right )^\klstar (d,e)\right )_\dr = d
  & = \left(\bigjoin\nolimits_i f_i^\klstar(d,e)\right)_\dr.
\end{flalign*}
Next we calculate the domains as follows:
\begin{flalign*}
  \dom\ \left ( \bigjoin\nolimits_i f_i \right )^\klstar (d,e)
  & = \{ t \leq d \mid {(f_k)}_\ev^0\comp e^t \dar \text{ for some } k \in \omega \} \\
  & = \bigcup\nolimits_i \{ t \leq d \mid (f_i)_\ev^0\comp e^t \dar \} \\
  & = \dom \left(\bigjoin\nolimits_i f_i^\klstar (d,e)\right).
\end{flalign*}
This yields~\eqref{eq:kl-cont2} straightforwardly. Let us now stick to the remaining option that $d \in \Rz$ and $e^d \dar$.
For the durations we have
\begin{flalign*}
  \left ( \left(\bigjoin\nolimits_i f_i \right)^\klstar (d,e) \right )_\dr
  & = d  + \sup\nolimits_i d_i \\
  & = \sup\nolimits_i (d + d_i)\\
  & = \left (\bigjoin\nolimits_i f_i^\klstar (d,e) \right )_\dr.
\end{flalign*}
Regarding domains, we calculate
\begin{flalign*}
  \dom\ \left ( \left (\bigjoin\nolimits_i f_i \right )^\klstar (d,e) \right )
  & = \left \{ t \leq d \mid (f_k)_\ev^0\comp e^t \dar \text{ for some } k \in \omega \right \}\\
    &\qquad \cup \left \{ t + d \in \Rz \mid \left(\bigjoin\nolimits_i f_i\right)_\ev^t e^d \dar \right \}  \\
  & = \bigcup\nolimits_i \left \{ t \leq d \mid (f_i)_\ev^0\comp e^t \dar \right \}
   \cup \bigcup\nolimits_i \left \{ t + d \in \Rz \mid (f_i)_\ev^t\comp e^d \dar \right \} \\
  & = \bigcup\nolimits_i \left \{ t \leq d \mid (f_i)_\ev^0\comp e^t \dar \right \} \cup
  \left \{ t + d \in \Rz \mid (f_i)_\ev^t \comp e^d\dar \right \} \\
  & = \dom\left(\bigjoin\nolimits_i f_i^\klstar (d,e) \right).
\end{flalign*}
Again, the equation~\eqref{eq:kl-cont2} is obtained straightforwardly.

Finally, let us check Clause~\ref{item:H0M3}, i.e.\ that $f^\klstar (0,\bot)=
(0,\bot)$ for every map $f:X\to H_0MY$. Indeed, the duration part of $f^\klstar
(0,\bot)$ is $0$ for $\bot$ is the totally undefined function, in particular,
undefined at $0$. The evolution part of $f^\klstar (0,\bot)$ is $\bot$ per
definition.
\qed

\subsection{Proof of Theorem~\ref{thm:retract}}
As an preparatory step, we prove two lemmas.
\begin{lemma}\label{lem:traj}
  Consider a map $f:X\to H_0 M (Y+X)$ and an element $x \in X$. The
  condition,
  \begin{flalign*}
    (f^\iistar(x))^t_\ev = y \in Y
  \end{flalign*}
  holds iff there exists a natural number $n \in \nats$ such that
  $(f^{\brks{n}}(x))^t_\ev = y \in Y$.
\end{lemma}
\begin{proof}
Suppose that $(f^{\brks{n}}(x))^t_\ev = y \in Y$ for some $n$ and note that 
$f^{\brks{n}}(x)\appr f^{\iistar}(x)$, which easily follows by induction on $n$.
By definition of the order $(f^{\brks{n}}(x))^t_\ev$ is defined and equals~$y$.
Suppose that conversely, $(f^\iistar(x))^t_\ev=y\in Y$. If $(f^{\brks{n}}(x))^t_\ev = y' 
\in Y$ for some $n$ and $y'$ then $y'=y$ by the previous argument and we are done.
Otherwise, $(f^{\brks{n}}(x))^t_\ev\uar$ for every $n$ which contradicts to the fact 
that $f^{\iistar}(x)$ is the least upper bound of all the $(f^{\brks{n}}(x))^t_\ev\uar$.
\end{proof}

\begin{lemma}\label{lem:itercong}
  Consider a natural number $n \in \nats$ and a non-negative real
  number $t \in \Rz$. If $(f^{\brks{n}} (x))^t_\ev \, \dar$ and
  $((\upsilon \rho f)^{\brks{n}} (x))^t_\ev \uar$ then there exists a
  non-negative real number $t' \leq t$ such that
  $(f^{\brks{m}} (x))^{t'}_\ev \uar$ for all $m \geq n$.
\end{lemma}

\begin{proof}
  The proof follows by induction over $n$. The base case
  $n=0$ is vacuously true, because $(f^{\brks{n}} (x))_\ev=\bot$ and therefore 
  the premise $(f^{\brks{n}} (x))^t_\ev \, \dar$ is not true for any $t$.
   
  For the induction step assume that $((\upsilon \rho f)^{\brks{n + 1}}
  (x))^t_\ev \uar$ and $(f^{\brks{n + 1}} (x))^t_\ev\dar$. By
  definition of $(-)^{\brks{n + 1}}$, equivalently, $([\eta,(\upsilon
  \rho f)^{\brks{n}}]^\star \upsilon \rho f (x))^t_\ev \uar$ and
  $([\eta,f^{\brks{n}}]^\star f (x))^t_\ev\dar$. We proceed by case
  distinction. If $t \leq f_\dr(x)$ then the assumption $([\eta,(\upsilon \rho
  f)^{\brks{n}}]^\star \upsilon \rho f (x))^t_\ev \uar$ is equivalent to
  $([\eta,f^{\brks{n}}]^\star \upsilon \rho f (x))^t_\ev \uar$, using the easily 
  verified fact that $((\upsilon \rho f)^{\brks{n}})_\ev^0 = (f^{\brks{n}})_\ev^0$. 
  Since by another assumption $([\eta,f^{\brks{n}}]^\star f (x))^t_\ev\dar$, there 
  exists a time instant $t' \leq t$ such that $(f(x))^{t'}_\ev \uar$. Since $t' 
  \leq t \leq f_\dr(x)$, clearly, $(f^{\brks{m}} (x))^{t'}_\ev \uar$ for all $m\geq n$
  and we are done. 

  In the remaining case $t > f_\dr(x)$, either 
  $(\upsilon \rho f (x))^{f_\dr(x)}_\ev = \inr x'$ and 
  $((\upsilon\rho f)^{\brks{n}}(x'))_\ev^{t - f_\dr(x)}\uar$ and 
  $(f^{\brks{n}}(x'))_\ev^{t - f_\dr(x)}\dar$ and we reduce to the induction hypothesis,
  or $(\upsilon \rho f (x))^{f_\dr(x)}_\ev\uar$, which implies $(\upsilon \rho f (x))^{t'}_\ev\uar$
  for some $t' < f_\dr(x) < t$. In the latter case $(f^{\brks{m}} (x))^{t'}_\ev \uar$, as desired.
\end{proof}
Let us continue the proof of Theorem~\ref{thm:retract}.
In order to show equality of evolutions, we reason as follows:
    \begin{flalign*}
      &&  (\rho (\upsilon \rho f)^\iistar&(x))^t_{\ev} = y_t \in Y\\
      && \Rightarrow & ~
      \forall t' \leq t. \> ((\upsilon \rho f)^\iistar(x))^{t'}_{\ev} = y_{t'} \in Y
      & \by{Definition of $\rho$} \\
      && \Rightarrow & ~ \forall t' \leq t. \> \exists n_{t'} \in \nats. \>
      ((\upsilon \rho f )^{\brks{n_{t'}}}(x))^{t'}_{\ev} = y_{t'} \in Y &
    \by{Lemma~\ref{lem:traj}}\\
    && \Rightarrow & ~ \forall t' \leq t. \> \exists n_{t'} \in \nats. \>
    (f^{\brks{n_{t'}}}(x))^{t'}_{\ev} = y_{t'} \in Y & \\
    && \Rightarrow & ~ (\rho f^\iistar (x))^t_{\ev} = y_t \in Y &
    \by{Lemma~\ref{lem:traj}, defn. of $\rho$}
  \intertext{Conversely,}
      && (\rho f^\iistar(x))^t_{\ev}& = y_t \in Y\\
      && \Rightarrow & ~
      \forall t' \leq t. \> (f^\iistar(x))^{t'}_{\ev} = y_{t'} \in Y
      & \by{Definition of $\rho$} \\
      && \Rightarrow & \> \forall t' \leq t. ~ \exists n_{t'} \in
      \nats. \> (f^{\brks{n_{t'}}} (x))^{t'}_{\ev} = y_{t'} \in Y &
      \by{Lemma~\ref{lem:traj}}\\
      && \Rightarrow & ~ \forall t' \leq t. \> \exists n_{t'} \in
      \nats. \>
      ((\upsilon \rho f)^{\brks{n_{t'}}}(x))^{t'}_{\ev} = y_{t'} \in Y &
      \by{Lemma~\ref{lem:itercong}}\\
      && \Rightarrow & ~ (\rho (\upsilon \rho f)^\iistar (x))^t_{\ev} = y_t \in Y &
      \by{Lemma~\ref{lem:traj}, defn. of $\rho$}
    \end{flalign*}
    Next we will show that the trajectories $\rho f^\iistar(x)$ and
    $\rho (\upsilon \rho f)^\iistar(x)$ have the same duration.

Suppose that the
trajectory $f^\iistar(x)$ is total. This means that the trajectory
$(\upsilon\rho f)^\iistar(x)$ must also be total, and therefore the
equation that we want to prove reduces to
$f_\dr^\iistar(x) = ((\upsilon\rho f)^\iistar(x))_\dr$. We are thus
left to check that
$((\upsilon\rho f)^{\brks{i}}(x))_\dr=(f^{\brks{i}}(x))_\dr$ for every
$i$, which is straightforward by induction over $i$. Now suppose that
the trajectory $f^\iistar(x)$ is not total. This means that the
trajectory $(\upsilon\rho f)^\iistar(x)$ will also not be total and
therefore
$(\rho f^\iistar(x))_\dr = \infty = (\rho (\upsilon\rho
f)^\iistar(x))_\dr$.
\qed

\end{document}